\documentclass[submission,copyright,creativecommons]{eptcs}
\providecommand{\event}{GCM 2019} 
\usepackage{breakurl}             

\usepackage{amsmath,amssymb,amsthm}
\usepackage{bbold}
\usepackage{mathtools}
\usepackage{etoolbox}
\usepackage{xifthen}

\usepackage{array}
\usepackage{xcolor}
\usepackage{wrapfig}
\usepackage{graphicx}

\usepackage{pdflscape}
\usepackage{afterpage}

\usepackage{thmtools}

\usepackage{thm-restate}

\theoremstyle{plain}
\newtheorem{theorem}{Theorem}
\newtheorem{lemma}{Lemma}

 \newtheorem{asmptn}{Assumption}

 \theoremstyle{definition}
 \newtheorem{definition}{Definition}
 \newtheorem{example}{Example}

\usepackage{hyperref}

\makeatletter
\def\bstr{b}
\def\bfstr{bf}
\def\cstr{c}
\def\fstr{f}
\def\lst{A,B,C,D,d,E,F,G,H,I,J,K,L,M,N,O,P,Q,R,S,T,U,V,W,X,Y,Z}

\newcommand{\MkB}[1]{\expandafter\def\csname\bstr#1\endcsname{\mathbb{#1}}}
  \@for\i:=\lst\do{%
    \expandafter\MkB \i     }
    
\newcommand{\MkBF}[1]{\expandafter\def\csname\bfstr#1\endcsname{\mathbf{#1}}}
  \@for\i:=\lst\do{%
    \expandafter\MkBF \i     }

\newcommand{\MkCal}[1]{\expandafter\def\csname\cstr#1\endcsname{\mathcal{#1}}}
  \@for\i:=\lst\do{%
    \expandafter\MkCal \i     }
    
\newcommand{\MkFrak}[1]{\expandafter\def\csname\fstr#1\endcsname{\mathfrak{#1}}}
  \@for\i:=\lst\do{%
    \expandafter\MkFrak \i     }
    
\makeatother

\newcommand{\pB}[1]{\mathsf{PB}(#1)}
\newcommand{\pO}[1]{\mathsf{PO}(#1)}

\newcommand{\inObj}{\varnothing}

\newcommand{\mono}[1]{\mathsf{mono}(#1)}

\newcommand{\mor}[1]{\mathsf{mor}(#1)}

\newcommand{\obj}[1]{\mathsf{obj}(#1)}

\newcommand{\Lin}[1]{\mathsf{Lin}(#1)}

\newcommand{\sqMatch}[2]{\mathsf{M}^{sq}_{#1}(#2)}

\newcommand{\sqRMatch}[2]{\mathbf{M}^{sq}_{#1}(#2)}

\newcommand{\cSquare}[1]{\Box(#1)}

\newcommand{\sqComp}[3]{#1 \stackrel{#2}{\sphericalangle} #3}

\newcommand{\grule}[3]{#1\stackrel{#2}{\Leftarrow}#3}

\newcommand{\bra}[1]{\left\langle #1\right\vert}
\newcommand{\ket}[1]{\left\vert #1\right\rangle}
\newcommand{\braket}[2]{\left\langle \left. #1 \right\vert #2\right\rangle}

\gdef\tpScale{0.7}

\newcommand{\tP}[1]{\ensuremath{\vcenter{\hbox{\begin{tikzpicture}[transform shape, scale=\tpScale]\normalsize#1\end{tikzpicture}}}}}

\newcommand{\tPL}[1]{\begin{tikzpicture}[transform shape, scale=\tpScale]\normalsize#1\end{tikzpicture}}

\gdef\lblSEP{0.2} 

\colorlet{grey}{black!20!white}

\colorlet{h1color}{blue!40!black} 
\colorlet{h2color}{orange!90!black} 
\colorlet{h3color}{blue!40!white} 
\colorlet{h4color}{green!40!black} 

\newcommand{\OneVertG}[1][]{%
\ifthenelse{\equal{#1}{}}{%
\tP{%
\node[vertices] (a) at (1,1) {};}}{%
\tP{%
\node[vertices,#1] (a) at (1,1) {};}}}

\newcommand{\TwoVertG}[1][]{%
\ifthenelse{\equal{#1}{}}{%
\tP{%
\node[vertices] (a) at (1,1) {};
\node[vertices] (b) at (1.7,1) {};}}{%
\tP{%
\node[vertices,#1] (a) at (1,1) {};
\node[vertices,#1] (b) at (1.7,1) {};}}}

\newcommand{\OneVertGL}[2][]{%
\ifthenelse{\equal{#1}{}}{
\raisebox{-0.3em}{\tPL{
\node[vertices] (a) at (1,1) {};
\node at ($(a.south) -(0,\lblSEP)$) {{\tiny $#2$}};}}}{%
\raisebox{-0.3em}{\tPL{%
\node[vertices,#1] (a) at (1,1) {#2};
\node[#1] at ($(a.south) -(0,\lblSEP)$) {{\tiny $#2$}};}}}}

\newcommand{\TwoVertGL}[3][]{%
\ifthenelse{\equal{#1}{}}{%
\tP{%
\node[vertices] (a) at (1,1) {#2};
\node[vertices] (b) at (1.5,1) {#3};}}{%
\tP{%
\node[vertices,#1] (a) at (1,1) {#2};
\node[vertices,#1] (b) at (1.5,1) {#3};}}}

\newcommand{\TwoVertEdgeG}[1][]{%
\ifthenelse{\equal{#1}{}}{\tP{%
\node[vertices] (a) at (1,1) {};
\node[vertices] (b) at (1.7,1) {};
\draw (a) edge (b);}}{\tP{%
\node[vertices,#1] (a) at (1,1) {};
\node[vertices,#1] (b) at (1.7,1) {};
\draw (a) edge[#1] (b);}}}

\newcommand{\TwoVertDirEdgeG}[1][]{%
\ifthenelse{\equal{#1}{}}{\tP{%
\node[vertices] (a) at (1,1) {};
\node[vertices] (b) at (1.7,1) {};
\draw (a) edge[dirEdge] (b);}}{\tP{%
\node[vertices,#1] (a) at (1,1) {};
\node[vertices,#1] (b) at (1.7,1) {};
\draw (a) edge[dirEdge,#1] (b);}}}

\newcommand{\TwoVertDirEdgeGLb}[3]{%
\raisebox{-0.3em}{\tPL{%
\node[vertices,#1] (a) at (1,1) {};
\node[vertices,#1] (b) at (1.7,1) {};
\node at ($(a.south) -(0,\lblSEP)$) {{\tiny $#2$}};
\node at ($(b.south) -(0,\lblSEP)$) {{\tiny $#3$}};
\draw (a) edge[dirEdge,#1] (b);}}}

\gdef\mycdScale{1}

\usepackage{environ}

\def\temp{&} \catcode`&=\active \let&=\temp

\NewEnviron{mycd}[1][]{%
\begin{tikzpicture}[baseline=(current bounding box.west),
transform shape, scale=\mycdScale]
\node[inner sep=0, outer sep=0] {\begin{tikzcd}[#1]
\BODY
\end{tikzcd}};
\end{tikzpicture}}

 \usepackage{tikzexternal}

\tikzsetexternalprefix{tikz/}
\usepackage[strings]{underscore}

 \title{Sesqui-Pushout Rewriting:\\ Concurrency, Associativity and Rule Algebra Framework\thanks{This project has received funding from the European Union's Horizon 2020 research and innovation programme under the Marie Sk\l{}odowska-Curie grant agreement No~753750.}}
\author{Nicolas Behr
\institute{Universit{\'{e}} de Paris, IRIF, CNRS\\
F-75013 Paris, France}
\email{nicolas.behr@irif.fr}%
}

\def\titlerunning{SqPO Rewriting: Concurrency, Associativity and Rule Algebra Framework}
\def\authorrunning{N. Behr}

\begin{document}
\maketitle              
\begin{abstract}
Sesqui-pushout (SqPO) rewriting is a variant of transformations of graph-like and other types of structures that fit into the framework of adhesive categories where deletion in unknown context may be implemented. We provide the first account of a concurrency theorem for this important type of rewriting, and we demonstrate the additional mathematical property of a form of associativity for these theories. Associativity may then be exploited to construct so-called rule algebras (of SqPO type), based upon which in particular a universal framework of continuous-time Markov chains for stochastic SqPO rewriting systems may be realized.
\end{abstract}
\section{Motivation and relation to previous work}

The framework of \emph{Sesqui-Pushout (SqPO) rewriting} has been introduced relatively recently in~\cite{Corradini_2006} as a novel alternative to the pre-existing algebraic graph transformation frameworks known as \emph{Double-Pushout (DPO)}~\cite{Ehrig1973,Ehrig1991,DBLP:conf/gg/CorradiniMREHL97,lack2005adhesive} and \emph{Single-Pushout (SPO) rewriting}~\cite{Kennaway,Loewe_1993,Loewe_2014}. In the setting of the rewriting of graph-like structures, the distinguishing feature of the aforementioned DPO-type rewriting is that the deletion of vertices with incident edges is only possible if the incident edges are explicitly deleted via the application of the rewriting rule. In contrast, in both the SqPO and the SPO rewriting setups, \emph{``deletion in unknown context''} is implementable. Thus for practical applications of rewriting, in particular in view of the modeling of stochastic rewriting systems, the S(q)PO rewriting semantics provide an important additional option for the practitioners, and will thus in particular complement the existing DPO-type associative rewriting and rule algebra framework as introduced in~\cite{bdg2016,bp2018}. Referring the interested readers to~\cite{Loewe_2015} for a recent review and further conceptual details of the three approaches, suffice it here to quote that SqPO and SPO rewriting via linear rules\footnote{While non-linear rules in SqPO rewriting have interesting applications in their own right (permitting e.g.\ the cloning and fusing of vertices in graphs), this most general case is left for future work.} (defined as monic spans) and along monomorphic matches effectively encode the same semantics of rewriting. We chose (by the preceding argument without loss of expressivity) to develop the theory of associative rewriting within the SqPO rather than the SPO setting, since the SqPO framework bears certain close technical similarities to the DPO rewriting framework, which proved crucial in finding a strategy for the highly intricate proofs of the concurrency and associativity theorems presented in this paper.  While it is well-known (see e.g.\ Section~5.1 of~\cite{Corradini_2006}) that DPO- and SqPO-type semantics coincide for certain special classes of linear rules (essentially rules that do not delete vertices), and while these cases might provide some valuable cross-checks of technical results to the experts, SqPO-type semantics is in its full generality a considerably more intricate variant of semantics due to its inherent ``mixing'' of pushouts with final pullback complements. It should further be noted that we must impose a set of additional assumptions on the underlying adhesive categories (see Assumption~\ref{ass:SqPO}) in order to ensure certain technical properties necessary for our concurrency and associativity theorems to hold. To the best of our knowledge, apart from some partial results in the direction of developing a concurrency theorem for SqPO-type rewriting in~\cite{Corradini_2006,Loewe_2015,Corradini2018}, prior to this work neither of the aforementioned theorems had been available in the SqPO framework.\\

Associativity of SqPO rewriting theories plays a pivotal role in our  development of a novel form of concurrent semantics for these theories, the so-called \emph{SqPO-type rule algebras}. Previous work on associative DPO-type rewriting theories~\cite{bdg2016,bdgh2016,bp2018} (see also~\cite{bp2019-ext}) has led to a category-theoretical understanding of associativity that may be suitably extended to the SqPO setting. In contrast to the traditional and well-established formalisms of concurrency theory for rewriting systems (see e.g.~\cite{DBLP:conf/gg/1997handbook,ehrig2004adhesive,EHRIG:2014ma,Corradini2018} for DPO-type semantics and~\cite{Corradini_2006,Corradini2018} for a notion of parallel independence and a Local Church-Rosser theorem for SqPO-rewriting of graphs), wherein the focus of the analysis is mostly on \emph{derivation traces} and their sequential independence and parallelism properties, the focus of our rule-algebraic approach differs significantly: we propose instead to put \emph{sequential compositions of linear rules} at the center of the analysis (rather than the derivation traces), and moreover to employ a vector-space based semantics in order to encode the non-determinism of such rule compositions. It is for this reason that the concurrency theorem plays a quintessential role in our rule algebra framework, in that it encodes the relationship between sequential compositions of linear rules and derivation traces, which in turn gives rise to the so-called \emph{canonical representations} of the rule algebras (see Section~\ref{sec:ACDrd}). This approach in particular permits to uncover certain combinatorial properties of rewriting systems that would otherwise not be accessible. While undoubtedly not a standard technique in the realm of theoretical computer science, certain special examples of rule algebras are ubiquitous in many areas of applied mathematics and theoretical physics. The most famous such example concerns the so-called \emph{Heisenberg-Weyl algebra} (see e.g.\ \cite{blasiak2005boson,blasiak2010combinatorial,blasiak2011combinatorial}), which is well-known to possess a representation in terms of the formal multiplication operator $\hat{x}$ and the differentiation operator $\partial_x$ on formal power series in the formal variable $x$, with $\hat{x}\,x^n:=x^{n+1}$ and $\partial_x$ acting as the derivative. Referring the interested readers to Example~\ref{ex:repInt} (see also~\cite{bp2018,bdg2019}) for the precise details, it transpires that the monomials $x^n$ (for $n$ a non-negative integer) are found to be in one-to-one correspondence with \emph{graph states} associated to $n$-vertex discrete graphs, while $\hat{x}$ and $\partial_x$ may be understood as the \emph{canonical representations} of the discrete graph rewriting rules of creation and deletion of vertices. It will thus come as no surprise that considering more general rewriting rules than those of discrete graphs will lead to a very substantial generalization of these traditional results and techniques.\\

From the very beginning of the development of the rule algebra framework~\cite{bdg2016}, one of our main motivations has been the study of stochastic rewriting systems, whence of continuous-time Markov chains (CTMCs) based upon DPO- or SqPO-type rewriting rules. While previously in particular applications of stochastic SqPO-type rewriting systems have played a role predominantly in highly specialized settings such as e.g.\ the formulation of the biochemical reaction system framework known as \emph{\textsc{Kappa}}~\cite{Danos_2010,danos2004formal,danos2008rule,danos2012graphs}, our novel approach of formulating such systems in terms of associative unital rule algebras may very well open this versatile modeling technique to many other areas of applied research. In conjunction with our previously developed DPO-type framework in~\cite{bp2018}, one could argue that our \emph{stochastic mechanics frameworks} are in a certain sense a \emph{universal construction}, in that once a semantics for associative unital rewriting is provided, the steps necessary to obtain the associated CTMCs are clearly formalized. It is interesting to compare the traditional approaches to stochastic rewriting systems with CTMC semantics such as~\cite{Heckel_2004,Heckel2012} in the DPO- and~\cite{Heckel2005} in the SPO-settings to our present reformulation in terms of rule algebras. The former approaches yet again tend to focus on derivation traces of stochastic rewriting systems, while our rule-algebraic approach aims to extract dynamical information from stochastic rewriting systems via analysis of certain combinatorial relationships (so-called nested commutators) of the infinitesimal generator of the CTMC with the (pattern-counting) observables of the system. It is via these relations that one may in certain cases obtain \emph{exact closed-form solutions} for such dynamical data (see e.g.\ Section~\ref{sec:appEx}). It would nevertheless be an intriguing avenue for future research to understand better the finer points of the ``traditional'' stochastic rewriting frameworks (which also feature sophisticated developments in terms of probabilistic model-checking and various types of stochastic logics), and furthermore whether or not rule-algebraic techniques might be of interest also in more general stochastic rewriting semantics such as probabilistic (timed) graph transformations~\cite{Krause2012,Maximova2018}.\\

\noindent\textbf{Structure of the paper:}~In Section~\ref{sec:Adh}, some category-theoretical background material is provided. The key results of associativity and concurrency of SqPO rewriting are presented in Section~\ref{ec:SqPO}, followed by the construction of SqPO-type rule algebras in Section~\ref{sec:ACDrd}. The second part of the paper contains the stochastic mechanics framework (Section~\ref{sec:SM}) as well as a practical application example (Section~\ref{sec:appEx}). Technical proofs are situated in the Appendix.

\section{Background: adhesive categories and final pullback complements}
\label{sec:Adh}

We recall some of the elementary definitions and properties related to the notions of adhesive categories, upon which our framework will rely. 

\begin{definition}[\cite{lack2005adhesive}]\label{def:adhCats}
A category $\bfC$ is said to be \textbf{\emph{adhesive}} if 
\begin{enumerate}
\item $\bfC$ has pushouts along monomorphisms,
\item $\bfC$ has pullbacks,
\item pushouts along monomorphisms are van Kampen (VK) squares.
\end{enumerate}
The last property entails that in a commutative cube as in~\eqref{eq:diags} on the left where the bottom square is a pushout, this square is a VK square if and only if whenever the back and right vertical faces are pullbacks, then the top square is a pushout if and only if the front and left vertical squares are pullbacks.
\end{definition}

\begin{equation}\label{eq:diags}
\vcenter{\hbox{\includegraphics[scale=0.7]{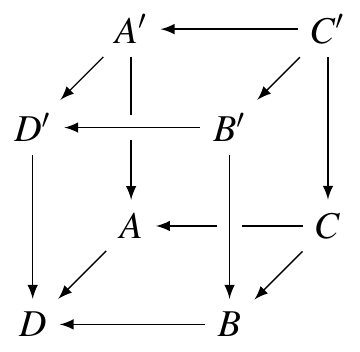}}}\qquad \quad
\vcenter{\hbox{\includegraphics[scale=0.8]{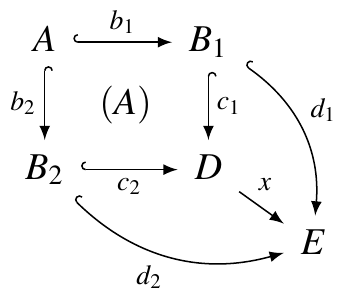}}}\qquad \quad \vcenter{\hbox{\includegraphics[scale=0.8]{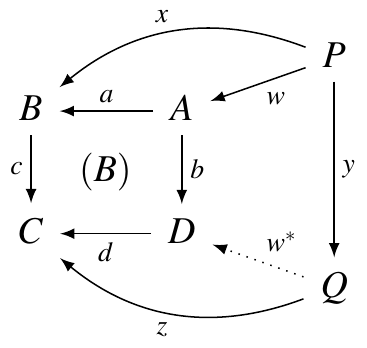}}}
\end{equation}
We will be exclusively interested in categories that satisfy certain finiteness properties (in order to ensure finiteness of the sets of matches for rule applications and compositions, see Section~\ref{ec:SqPO}):
\begin{definition}[Finitary categories~\cite{Braatz:2010aa}]
A category $\bfC$ is said to be \emph{finitary} if every object $X\in \obj{\bfC}$ has only finitely many subobjects (i.e.\ if there only exist finitely many monomorphisms $Y\rightarrow X$ up to isomorphism for every $X\in \obj{\bfC}$). For every adhesive category $\bfC$, the restriction to finite objects of $\bfC$ defines a full subcategory $\bfC_{fin}$ called the \emph{finitary restriction} of $\bfC$.
\end{definition}
\begin{theorem}[Finitary restrictions; \cite{Braatz:2010aa}, Thm.~4.6]
The \emph{finitary restriction} $\bfC_{fin}$ of any adhesive category $\bfC$ is a \emph{finitary adhesive category}.
\end{theorem}

Adhesive categories have been introduced and advocated in~\cite{lack2005adhesive} as a framework for rewriting due to their numerous useful properties, some of which are listed in Appendix~\ref{app:lemList} for the reader's convenience. One of the central concepts in the theory of SqPO rewriting is the following:
\begin{definition}[Final Pullback Complement (FPC); \cite{Corradini_2006,Loewe_2015}]
	Let $\bfC$ be a category. Given a commutative diagram as in~\eqref{eq:diags} on the right, a pair of morphisms $(d,b)$ is a \emph{final pullback complement (FPC)} of a pair $(c,a)$ if (i) $(a,b)$ is a pullback of $(c,d)$ (i.e.\ if the square marked $(B)$ is a pullback square), and (ii) for each collection of morphisms $(x, y, z, w)$ as in~\eqref{eq:diags} on the right, where $(x,y)$ is pullback of $(c, z)$ and where $a\circ w=x$, there exists a unique morphism $w^{*}$ with $d\circ w^{*}=z$ and $w^{*}\circ y=b\circ w$.
\end{definition}
 
For our associative rewriting framework, it will be crucial to work with a category in which (i) FPCs are guaranteed to exist when constructing them for composable pairs of monomorphisms, and (ii) monomorphisms are stable under FPCs, i.e.\ FPCs of pairs of monomorphisms are given by pairs of monomorphisms. This property is satisfied by adhesive categories (cf.\ Lemma~\ref{lem:FPCfacts} of Appendix~\ref{app:lemList}), yet to the best of our knowledge the question of which more general types of categories possess this property has not yet been investigated to quite the level of generality as analogous classification problems in the case of DPO rewriting, even though there does exist a large body of work on classes of categories that admit SqPO constructions~\cite{Corradini_2006,lowe2015single,Loewe_2015,Corradini_2015,Loewe_2018}. Within these classes, according to~\cite{Cockett_2003,Corradini_2015} guarantees for the existence of FPCs may be provided for categories that possess a so-called $\cM$-partial map classifier. However, it appears to be an open question of whether the statement of Lemma~\ref{lem:FPCfacts} on stability of monomorphisms under FPCs may be generalized to the setting of $\cM$-adhesive categories, where $\cM$ is a class of monomorphisms. Relaying such questions to future work, we refer to Lemma~\ref{lem:GraphFPC} of Appendix~\ref{app:lemList} for a well-known instantiation of a suitable categorical setting from the SqPO literature in the form of the category $\mathbf{FinGraph}$ of finite directed multigraphs, which  also serves to illustrate the FPC construction.

\begin{asmptn}\label{ass:SqPO}
	$\bfC$ is an adhesive category in which all FPCs along monomorphisms exist, and in which monomorphisms are stable under FPCs.
\end{asmptn}

\section{Sesqui-Pushout rewriting}
\label{ec:SqPO}

We will now develop a framework for \emph{Sesqui-Pushout (SqPO) rewriting} in the setting of a category $\bfC$ satisfying Assumption~\ref{ass:SqPO}, in close analogy to the framework of associative Double-Pushout (DPO) rewriting as introduced in~\cite{bp2018,bp2019-ext}. Unlike in the general setting of SqPO rewriting, we will thus be able to not only prove a \emph{concurrency theorem} (Section~\ref{sec:SqPOconc}), but also an \emph{associativity property} of the SqPO-type rule composition (Section~\ref{sec:SqPOassoc}).

\subsection{Concurrent composition and concurrency theorem}\label{sec:SqPOconc}

For reasons that will become more transparent when introducing the SqPO-type rule algebra framework starting from Section~\ref{sec:ACDrd}, we opt for a non-standard convention of reading spans of monomorphisms ``from right to left'' (rather than the traditional ``left to right''), which is why we will speak of ``input'' and ``output'' of rules rather than ``left-'' and ``right hand sides'' to avoid confusion.

\begin{definition}[SqPO-type rewriting; compare \cite{Corradini_2006}, Def.~4]
\label{def:SqPOr}
Let $\bfC$ be an adhesive category satisfying Assumption~\ref{ass:SqPO}. Denote by $\Lin{\bfC}$ the set of (isomorphism classes\footnote{Two productions $O\leftarrow K\rightarrow I$ and $O'\leftarrow K'\rightarrow I'$ are defined to be isomorphic if there exist isomorphisms $I\rightarrow I'$, $K\rightarrow K'$ and $O\rightarrow O'$ that make the obvious diagram commute; we will not distinguish between isomorphic productions. As natural in this category-theoretical setting, the constructions presented in the following are understood as defined up to such isomorphisms.} of) so-called \emph{linear productions}, defined as the set of spans of monomorphisms,
\begin{equation}
\Lin{\bfC}:=\{p\equiv (O\xleftarrow{o}K\xrightarrow{i}I)\mid o,i\in \mono{\bfC}\}\diagup_{\cong}\,.
\end{equation}
Given an object $X\in \obj{\bfC}$ and a linear production $p\in \Lin{\bfC}$, we denote the \emph{set of SqPO-admissible matches} $\sqMatch{p}{X}$ as the set of monomorphisms $m:I\rightarrow X$. Then the diagram below is constructed by taking the \emph{final pullback complement} marked $\mathsf{FPC}$ followed by taking the pushout marked $\mathsf{PO}$:
\begin{equation}\label{eq:DPOr}\gdef\mycdScale{0.85}
\begin{mycd}
O \ar[d,"{m^{*}}"'] & 
 K \ar[l,"o"']\ar[r,"i"]\ar[d,"k"']
		\ar[dl,phantom, "{\mathsf{PO}}"]\ar[dr,phantom,"{\mathsf{FPC}}"] &
 I \ar[d,"m"]
 \\
 {X'} & {\overline{K}} \ar[l,"o'"]\ar[r,"i'"'] & X\\
\end{mycd}
\end{equation}
We write $p_m(X):=X'$ for the object ``produced'' by the above diagram. The process is called \emph{(SqPO-) derivation} of $X$ along production $p$ and admissible match $m$, and denoted $p_m(X)\xLeftarrow[p,m]{{\tiny SqPO}} X$.
\end{definition}

Next, a notion of sequential composition of productions is introduced:

\begin{definition}[SqPO-type concurrent composition]\label{def:SqPOcomp}
	Let $p_1,p_2\in \Lin{\bfC}$ be two linear productions. Then an overlap of the output object $O_1$ of $p_1$ with the input object $I_2$ of $p_2$, encoded as a span 
	\[
		{\color{h1color}\mathbf{m}}=(I_2{\color{h1color}\xleftarrow{m_2} M_{21}\xrightarrow{m_1}}O_1)
	\] 
	with $m_1,m_2\in \mono{\bfC}$, is called an \emph{SqPO-admissible match of $p_2$ into $p_1$}, denoted $\mathbf{m}\in \sqMatch{p_2}{p_1}$, if the square marked $\mathsf{POC}$ in~\eqref{eq:SqPOccomp} is constructible as a pushout complement (with the cospan $I_2\xrightarrow{n_2}N_{21}\xleftarrow{n_1}O_1$ obtained by taking the pushout marked ${\color{h1color}\mathsf{PO}}$). In this case, the remaining parts of the diagram are formed by taking the final pullback complement marked $\mathsf{FPC}$ and the pushouts marked $\mathsf{PO}$:
	\begin{equation}\label{eq:SqPOccomp}
		\begin{mycd}
			O_2\ar[d,"n_2^{*}"'] & 
			K_2 \ar[l,"o_2"']\ar[r,"i_2"]\ar[d,"k_2"']
				\ar[dl,phantom,"\mathsf{PO}"] & 
			I_2 \ar[dr,h1color,bend right,"n_2"]\ar[dl,phantom,"\mathsf{FPC}"] & 
			{\color{h1color}M_{21}}
				\ar[l,h1color,"m_2"']\ar[r,h1color,"m_1"]\ar[d,h1color,phantom,"\mathsf{PO}"] & 
			O_1 \ar[dl,h1color,bend left,"n_1"']\ar[dr,phantom,"\mathsf{POC}"]&
			K_1 \ar[l,"o_1"']\ar[r,"i_1"]\ar[d,"k_1"]
				\ar[dr,phantom,"\mathsf{PO}"] & 
			I_1\ar[d,"n_1^{*}"]\\
			{\color{h2color}O_{21}} & 
			\overline{K}_2\ar[l,"o_2'"']\ar[rr,"i_2'"] & & 
			{\color{h1color}N_{21}} & &
			\overline{K}_1\ar[ll,"o_1'"']\ar[r,"i_1'"] & {\color{h2color}I_{21}}\\
			&&& {\color{h2color} K_{21}}
				\ar[u,phantom,h2color,"\mathsf{PB}"]
				\ar[ull,bend left=10,dotted,h2color,"i_2''"']
				\ar[ulll,bend left=10,h2color,near end,"o_{21}=o_2'\circ i_2''"]
				\ar[urr,bend right=10,dotted,h2color,"o_1''"]
				\ar[urrr,bend right=10,h2color,near end,"i_{21}=o_1''\circ i_1'"'] &&&
		\end{mycd}
	\end{equation}
	If $\mathbf{m}\in \sqMatch{p_2}{p_1}$, we write $\sqComp{p_2}{\mathbf{m}}{p_1}\in \Lin{\bfC}$ for the \emph{composite} of $p_2$ with $p_1$ along the admissible match $\mathbf{m}$, defined as
	\begin{equation}\label{eq:defSqPOcomp}
	\begin{aligned}
		\sqComp{p_2}{\mathbf{m}}{p_1}&\equiv {\color{h2color}(O_{21}\xleftarrow{o_{21}}K_{21}\xrightarrow{i_{21}}I_{21})}
		\,.
	\end{aligned}
	\end{equation}
\end{definition}
Due to stability of monomorphisms under pushouts, pullbacks and FPCs in the setting of a category satisfying Assumption~\ref{ass:SqPO}, all morphisms in Definitions~\ref{def:SqPOr} and~\ref{def:SqPOcomp} are guaranteed to be monomorphisms, whence in particular the span $\sqComp{p_2}{\mathbf{m}}{p_1}$ is a span of monomorphisms and thus indeed an element of $\Lin{\bfC}$.\\

At first sight, it might appear irritating that in the definition of the SqPO-type rule composition, the right hand part of~\eqref{eq:SqPOccomp} involves a \emph{pushout complement} (marked $\mathsf{POC}$), while the left hand part of the diagram in~\eqref{eq:SqPOccomp} features a \emph{final pullback complement} (marked $\mathsf{FPC}$). Intuitively, considering the case of graph rewriting for concreteness, in a given sequential application of two productions, while the application of the first production may lead to implicit edge deletions, the second production is incapable of having any causal interaction with edges deleted by the first production. In contrast, the second production may in a given sequential application very well implicitly delete edges present in the output object of the first production, which explains the presence of the FPC in the defining equation~\eqref{eq:SqPOccomp}. We refer the interested readers to~\cite{bdgh2016} for further intuitions attainable in terms of so-called rule diagrams for presenting rule compositions. The justification for Definition~\ref{def:SqPOcomp} in the general case is provided via the following \emph{concurrency theorem}. Even though at least in certain specialized settings the ``synthesis'' part of this theorem has been foreseen already in~\cite{Loewe_2015} (where it is also commented that a full concurrency theorem for SqPO rewriting might be attainable), the following result appears to be new.

\begin{theorem}[SqPO-type Concurrency Theorem]\label{thm:SqPOconcur}
Let $\bfC$ be an adhesive category satisfying Assumption~\ref{ass:SqPO}. Let $p_1,p_2\in \Lin{\bfC}$ be two linear rules and $X_0\in ob(\bfC)$ an object.
\begin{itemize}
\item \textbf{Synthesis:} Given a two-step sequence of SqPO derivations 
\[
X_2\xLeftarrow[p_2,m_2]{{\tiny SqPO}} X_1\xLeftarrow[p_1,m_1]{{\tiny SqPO}}X_0\,,
\]
with $X_1:=p_{1_{m_1}}(X_0)$ and $X_2:=p_{2_{m_2}}(X_1)$, there exists a SqPO-composite rule $q=\sqComp{p_2}{\mathbf{n}}{p_1}$
for a unique $\mathbf{n}\in \sqRMatch{p_2}{p_1}$,
 and a unique SqPO-admissible match $n\in \sqMatch{q}{X}$, such that 
 \[
 	q_n(X)\xLeftarrow[q,n]{{\tiny SqPO}} X_0\qquad \text{and}\qquad q_n(X_0)\cong X_2\,.
 \]
\item \textbf{Analysis:} Given an SqPO-admissible match $\mathbf{n}\in \sqRMatch{p_2}{p_1}$ of $p_2$ into $p_1$ and an SqPO-admissible match $n\in \sqMatch{q}{X}$ of the SqPO-composite $q=\sqComp{p_2}{\mathbf{n}}{p_1}$ into $X$, there exists a unique pair of SqPO-admissible matches $m_1\in \sqMatch{p_1}{X_0}$ and $m_2\in \sqMatch{p_2}{X_1}$ with $X_1:=p_{1_{m_1}}(X_0)$ such that
\[
		X_2\xLeftarrow[p_2,m_2]{{\tiny SqPO}} X_1 \xLeftarrow[p_1,m_1]{{\tiny SqPO}} X_0\qquad \text{and}\qquad
		X_2\cong q_n(X)\,.
\]
\end{itemize}
\begin{proof}
	See Appendix~\ref{app:SqPOconcur}.
\end{proof}
\end{theorem}

\subsection{Composition and associativity}\label{sec:SqPOassoc}

The following theorem establishes that in analogy to the DPO rewriting setting of~\cite{bp2018}, also the sesqui-pushout variant of rule compositions possesses a form of associativity property.

\begin{theorem}[SqPO-type associativity theorem]\label{thm:SqPOassoc}
	Let $\bfC$ be an adhesive category satisfying Assumption~\ref{ass:SqPO}. %
	Then the SqPO-composition operation $\sqComp{.}{.}{.}$ on linear productions of $\bfC$ is \emph{associative} in the following sense: %
given linear productions $p_1,p_2,p_3\in \Lin{\bfC}$, there exists a bijective correspondence
between pairs of SqPO-admissible matches $(\mathbf{m}_{21},\mathbf{m}_{3(21)})$ and $(\mathbf{m}_{32},\mathbf{m}_{(32)1})$ such that
\begin{equation}\label{eq:THMassoc}
	\sqComp{p_3}{\mathbf{m}_{3(21)}}{\left(\sqComp{p_2}{\mathbf{m}_{21}}{p_1}\right)}\; \cong \; 
	\sqComp{\left(\sqComp{p_3}{\mathbf{m}_{32}}{p_2}\right)}{\mathbf{m}_{(32)1}}{p_1}\,.
\end{equation}
\begin{proof}
Intuitively, the associativity property in the SqPO case manifests itself in a form entirely analogous to the DPO case~\cite{bp2018}, whereby the data provided along the path highlighted in orange below permits to uniquely compute the data provided along the path highlighted in blue and vice versa (with both sets of overlaps computing the same ``triple composite'' production that is encoded as the composition of the three spans in the bottom front row):
\begin{equation}
\vcenter{\hbox{\includegraphics[scale=0.4,page=1]{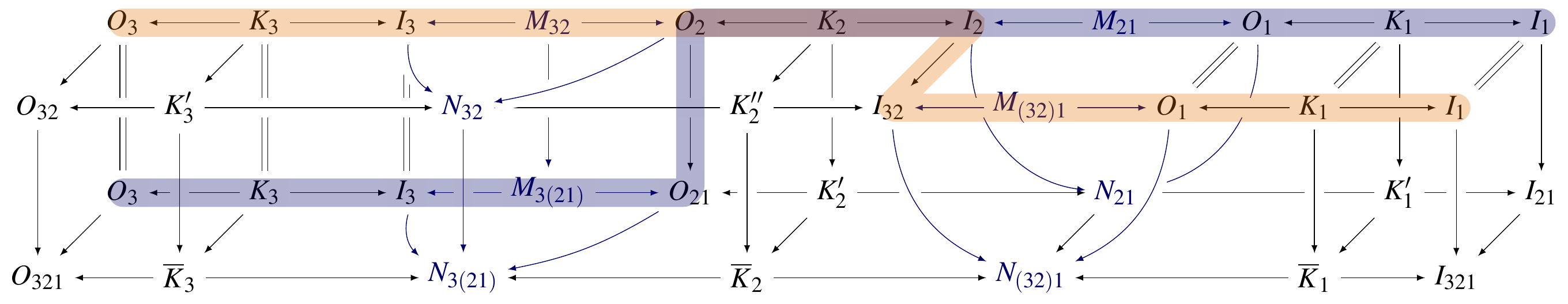}}}
\end{equation}
See Appendix~\ref{app:SqPOassoc} for the precise technical details of the proof.
\end{proof}
\end{theorem}

We invite the interested readers to compare the SqPO-type constructions presented here against those contained in the extended journal version~\cite{bp2019-ext} of~\cite{bp2018} for the DPO framework, since this might lend some intuitions on the otherwise very abstract nature of the proofs to the experts.


\section{From associativity to SqPO-type rule algebras}
\label{sec:ACDrd}


For the rule algebra constructions, we will require an additional structure:
\begin{definition}[Initial objects]
	An object $\inObj\in \obj{\bfC}$ of some category $\bfC$ is said to be a \emph{strict initial object} if for every object $X\in \obj{\bfC}$, there exists a unique morphism $\inObj\rightarrow X$, and if any morphism $X\rightarrow \inObj$ must be an isomorphism.
\end{definition}
For example, the category $\mathbf{Graph}$ and its finitary restriction $\mathbf{FinGraph}$ possess a strict initial object (the empty graph). For the experts, it appears worthwhile noting the following result:

\begin{lemma}[Extensive categories; \cite{lack2005adhesive}, Lem.~4.1]
	An adhesive category $\bfC$ is an \emph{extensive category}\footnote{For the purposes of this paper, it suffices to consider the ``if'' direction as a definition of extensivity, since the relevant structure to our constructions is that of having a strict initial object (see e.g.~\cite{lack2005adhesive} for the precise definition of extensivity).} if and only if it possesses a strict initial object.
\end{lemma}


\begin{asmptn}[Prerequisites for SqPO-type rule algebras]\label{ass:RAsqpo}
	We assume that $\bfC$ is an adhesive category satisfying Assumption~\ref{ass:SqPO}, and which is in addition \emph{finitary} and possesses a strict initial object $\inObj\in \obj{\bfC}$.
\end{asmptn}

\begin{definition}[SqPO-type rule algebras]
Let $\delta:\Lin{\bfC}\rightarrow \cR_{\bfC}$ be defined as an isomorphism from $\Lin{\bfC}$ to the basis of a free $\bR$-vector space $\cR_{\bfC}\equiv(\cR_{\bfC},+,\cdot)$, such that\footnote{Recall that for a set $A$, the notation $span_{\bR}(\{e(a)\mid a\in A\})$ entails to ``take the $\bR$-span over basis vectors $e(a)$ indexed by elements of $A$'', i.e.\ elements of the resulting $\bR$-vector space are (finite) linear combinations of the basis vectors $e(a)$ with real coefficients .}
\begin{equation}
\cR_{\bfC}:=span_{\bR}(\{\delta(p)\mid p\in \Lin{\bfC}\})\,.
\end{equation}
In order to clearly distinguish between elements of $\Lin{\bfC}$ and basis vectors of $\cR_{\bfC}$, we introduce the notation
\begin{equation}
(\grule{O}{p}{I}):=\delta\left(O\xleftarrow{o}K\xrightarrow{i}I\right)\,.
\end{equation}
Define the \emph{SqPO rule algebra product} $\odot_{\cR_{\bfC}}$ on a category $\bfC$ that satisfies Assumption~\ref{ass:RAsqpo} as the binary operation
\begin{equation}
\odot_{\cR_{\bfC}}:\cR_{\bfC}\times \cR_{\bfC}\rightarrow \cR_{\bfC}:(R_2,R_1)\mapsto R_2\odot_{\cR_{\bfC}} R_1\,,
\end{equation}
where for two basis vectors $R_i=\delta(p_i)$ encoding the linear rules $p_i\in Lin(\bfC)$ ($i=1,2$),
\begin{equation}\label{eq:defRcompSqPO}
R_2\odot_{\cR_{\bfC}}R_1
:=\sum_{\mathbf{m}\in \sqRMatch{p_2}{p_1}}\delta\left(\sqComp{p_2}{\mathbf{m}}{p_1}\right)\,.
\end{equation}
The definition is extended to arbitrary (finite) linear combinations of basis vectors by bilinearity, whence for $p_i,p_j\in \Lin{\bfC}$ and $\alpha_i,\beta_j\in \bR$,
\begin{equation}
\left(\sum_i \alpha_i\cdot\delta(p_i)\right)\odot_{\cR_{\bfC}}\left(\sum_j\beta_j\cdot \delta(p_j)\right):=\sum_{i,j}(\alpha_i\cdot\beta_j)\cdot \left(\delta(p_i)\odot_{\cR_{\bfC}}\delta(p_j)\right)\,.
\end{equation}
We call $\cR^{sq}_{\bfC}\equiv(\cR_{\bfC},\odot_{\cR_{\bfC}})$ the \textbf{\emph{SqPO-type rule algebra}} over the finitary adhesive and extensive category $\bfC$.
\end{definition}

The rule algebra product $R_2\odot_{\cR_{\bfC}}R_1$ for $R_j=\delta(r_j)$ ($j=1,2$) thus encodes the non-determinism in the SqPO-type sequential composition of the linear rule $r_2$ with $r_1$ in terms of the ``sum over all possible compositions''. As the following example illustrates, since $\delta$ is defined to map from \emph{isomorphism classes} of linear rules to basis vectors of $\cR_{\bfC}$, and since two distinct matches may lead to isomorphic composite rules, $R_2\odot_{\cR_{\bfC}}R_1$ typically evaluates to a linear combination of basis vectors $\delta(r)$ with integer coefficients:

\begin{example}\label{ex:ruleAlgComp}
	Let $\bfC=\mathbf{FinGraph}$ be the category of finite directed multigraphs, with $\inObj$ the empty graph. Then with $\odot\equiv \odot_{\cR_{\bfC}}$, we find for example 
	\begin{equation}
	\begin{aligned}
		&\delta(\inObj\hookleftarrow \inObj\hookrightarrow \OneVertG[])\odot
		\delta(\OneVertG[]\;\;\OneVertG[]\hookleftarrow \inObj\hookrightarrow\inObj)\\
		&\qquad=\sum_{\substack{\mathbf{m}\in \{
			(\OneVertG[]\hookleftarrow \inObj\hookrightarrow \OneVertG[]\;\;\OneVertG[]),
			(\OneVertG[blue]\hookleftarrow \OneVertG[blue]\hookrightarrow \OneVertG[blue]\;\;\OneVertG[]),\\
			\qquad(\OneVertG[blue]\hookleftarrow \OneVertG[blue]\hookrightarrow \OneVertG[]\;\;\OneVertG[blue])
		\}}}\delta\left(\sqComp{(\inObj\hookleftarrow \inObj\hookrightarrow \OneVertG[])}{\mathbf{m}}{(\OneVertG[]\;\;\OneVertG[]\hookleftarrow \inObj\hookrightarrow\inObj)}\right)\\
		&\qquad =
		\delta(\OneVertG[]\;\;\OneVertG[]\hookleftarrow \inObj\hookrightarrow \OneVertG[])
		+2\delta(\OneVertG[]\hookleftarrow \inObj\hookrightarrow \inObj)\,.
	\end{aligned}
	\end{equation}
	The result of the composition thus captures the combinatorial insight that there are two contributions that evaluate to an isomorphic rule algebra element. More generally, one finds the following structure of compositions of rule algebra elements based upon ``discrete'' graph rewriting rules: letting $\bullet^{\uplus\:n}$ denote the $n$-vertex graph without edges (for $n\geq 0$), one finds (for $p,q,r,s\geq 0$)
	\begin{equation}\label{eq:combinatorial}
	\begin{aligned}
	&\delta(\bullet^{\uplus\:p}\hookleftarrow \inObj\hookrightarrow \bullet^{\uplus\:q})
	\odot\delta(\bullet^{\uplus\:r}\hookleftarrow \inObj\hookrightarrow \bullet^{\uplus\:s})\\
	&\qquad =\sum_{k=0}^{min(q,r)}k!\binom{q}{k}\binom{r}{k}
	\delta(\bullet^{\uplus\:(p+r-k)}\hookleftarrow \inObj\hookrightarrow \bullet^{\uplus\:(q+s-k)})\,.
	\end{aligned}
	\end{equation}
	This result is further interpreted in Example~\ref{ex:repInt}.
\end{example}

\begin{theorem}[Properties of $\cR^{sq}_{\bfC}$]
For every category $\bfC$ satisfying Assumption~\ref{ass:RAsqpo}, the associated SqPO-type rule algebra $\cR^{sq}_{\bfC}\equiv(\cR_{\bfC},\odot_{\cR_{\bfC}})$ is an \emph{associative unital algebra}, with unit element $R_{\inObj}:=(\grule{\inObj}{}{\inObj})$. (Proof: Appendix~\ref{app:SqPOraProps})
\end{theorem}

For the unital and associative SqPO-type rule algebras, one may provide a notion of \emph{representations} in analogy to the DPO-type case (compare~\cite{bdg2016,bp2018}):
\begin{definition}[Canonical representation of $\cR^{sq}_{\bfC}$]\label{def:canRepSqPO}
Let $\bfC$ be a category satisfying Assumption~\ref{ass:RAsqpo}, with a strict initial object $\inObj\in ob(\bfC)$, and let $\cR^{sq}_{\bfC}$ be its associated rule algebra of SqPO type. Denote by $\hat{\bfC}$ the free $\bR$-vector space spanned by basis vectors $\ket{X}$ indexed by isomorphism classes of objects,
\begin{equation}
\hat{\bfC}:= span_{\bR}\left(\left\{\left.
\ket{X}
\right\vert X\in \obj{\bfC}_{\cong}
\right\}\right)\equiv (\hat{C},+,\cdot)\,.
\end{equation}
Then the \emph{canonical representation} $\rho^{sq}_{\bfC}:\cR^{sq}_{\bfC}\rightarrow End_{\bR}(\hat{\bfC})$ of $\cR^{sq}_{\bfC}$ is defined as a morphism from the SqPO-type rule algebra $\cR^{sq}_{\bfC}$ to endomorphisms of $\hat{\bfC}$, with
\begin{equation}\label{eq:canRepSqPO}
\rho^{sq}_{\bfC}(\delta(p))\ket{X}:=\begin{cases}
\sum_{m\in \sqMatch{p}{X}}\ket{p_m(X)}\quad &\text{if }\sqMatch{p}{X}\neq \emptyset\\
0_{\hat{\bfC}}&\text{otherwise,}
\end{cases}
\end{equation}
and extended to arbitrary elements of $\cR^{sq}_{\bfC}$ and of $\hat{\bfC}$ by linearity.
\end{definition}
\begin{example}\label{ex:repInt}
Extending Example~\ref{ex:ruleAlgComp}, letting $\rho\equiv \rho^{sq}_{\mathbf{FinGraph}}$, note first that by definition for all (isomorphism classes of) finite multigraphs $G\in \obj{\mathbf{FinGraph}}_{\cong}$, $\ket{G}=\rho(\delta(G\hookleftarrow \inObj\hookrightarrow \inObj))\ket{\inObj}$. With 
\begin{equation}\label{eq:HWrep}
\hat{D}:=\rho(\delta(\inObj\hookleftarrow \inObj\hookrightarrow \OneVertG[]))\,,\;
\hat{X}:=\rho(\delta(\OneVertG[]\hookleftarrow \inObj\hookrightarrow \inObj))\,,\;
\ket{n}:=\ket{\bullet^{\uplus\:n}} (n\geq0)\,,
\end{equation}
as a consequence of~\eqref{eq:combinatorial} of Example~\ref{ex:ruleAlgComp} one may verify that
\begin{equation}\label{eq:HWreP2}
\hat{D}\ket{0}=0_{\widehat{\mathbf{FinGraph}}}\,,\; \hat{D}\ket{n}=n\ket{n-1}\; (n>0)\,,\;
\hat{X}\ket{n}=\ket{n+1}\,.
\end{equation}
In other words, the data of~\eqref{eq:HWrep} and~\eqref{eq:HWreP2} furnishes a representation of the famous \emph{Heisenberg-Weyl algebra} that is of fundamental importance in combinatorics and physics (see e.g.\ \cite{blasiak2005boson,blasiak2010combinatorial,blasiak2011combinatorial}). An alternative such representation is given by the linear operators $\hat{x}$ (multiplication by $x$) and $\partial_x$ (derivation by $x$) acting on the $\bR$-vector space spanned by monomials $x^n$, which reproduces equations isomorphic to~\eqref{eq:HWrep} and~\eqref{eq:HWreP2}, with $\partial_x x^n=n x^{n-1}$ and $\hat{x} x^n=x^{n+1}$. However, the action of $\hat{D}$ and $\hat{X}$ is of course defined on \emph{all} states $\ket{G}$ with $G\in \obj{\mathbf{FinGraph}}$, so that we may e.g.\ compute the following ``derivative of a graph'':
\begin{equation}
\hat{D}\ket{\tP{%
\node[vertices] (a) at (1,1) {};
\node[vertices] (b) at (1.7,1) {};
\node[vertices] (c) at (2.4,1) {};
\draw (a) edge[dirEdge] (b);
\draw (b) edge[dirEdge] (c);}}=
2\ket{\tP{%
\node[vertices] (a) at (1,1) {};
\node[vertices] (b) at (1.7,1) {};
\draw (a) edge[dirEdge] (b);}}
+\ket{\tP{%
\node[vertices] (a) at (1,1) {};
\node[vertices] (b) at (1.7,1) {};}}
\end{equation}
\end{example}

The following theorem states that $\rho_C^{sq}$ as given in Definition~\ref{def:canRepSqPO} is indeed a homomorphism (and thus qualifies as a representation of $\cR^{sq}_{\bfC}$).
\begin{restatable}[SqPO-type canonical representation]{theorem}{canRepSqpo}
For a category $\bfC$ satisfying Assumption~\ref{ass:RAsqpo}, $\rho^{sq}_{\bfC}: \cR^{sq}_{\bfC} \rightarrow End(\hat{\bfC})$ of Definition~\ref{def:canRepSqPO} is a homomorphism of unital associative algebras. (Proof: Appendix~\ref{app:SqPOcanrep})
\end{restatable}


\section{Applications of SqPO-type rule algebras to stochastic mechanics}
\label{sec:SM}


In practical applications of stochastic rewriting systems, the type of rewriting semantics presents one of the key design choices. For example, if in a given situation a stochastic graph rewriting system should be implemented, choosing DPO- vs.\ SqPO-type rewriting entails two entirely different semantics in terms of the behavior of vertex deletion rules: in the former case, vertices may only be deleted if also all its incident edges are explicitly deleted as well, while in the latter case no such restriction applies (i.e.\ an application of a vertex deletion rule ``automatically'' leads to the deletion of all incident edges). Evidently, such fundamentally different behavior at the level of rewriting rules will also have strong influence on the dynamical behavior of the associated stochastic rewriting systems, whence it is of considerable practical interest to have a universal implementation of such systems available in both formalisms. We begin by specializing the general definition of continuous-time Markov chains (see e.g.\ \cite{norris,Anderson_1991}) to the setting of SqPO-type rewriting systems in close analogy to~\cite{bdg2016,bdp2017,bp2018}.

\begin{definition}[Continuous-time Markov Chains (CTMCs); compare \cite{bp2018}, Def.~7.1]\label{def:CTMCs}
Let $\bfC$ be a category satisfying Assumption~\ref{ass:RAsqpo}, and which in addition possesses a \emph{countable} set of isomorphism classes of objects $\obj{\bfC}_{\cong}$. Let $\hat{\bfC}$ denote the free $\bR$-vector space introduced in Definition~\ref{def:canRepSqPO}. We define the space $Prob(\bfC)$ as the \emph{space of sub-probability distributions} in the following sense:
\begin{equation}
Prob(\bfC):=\left.\left\{
\ket{\Psi}=\!\!\!\!\!\!\sum_{o\in \obj{\bfC}_{\cong}}\!\!\!\!\!\!\psi_o \ket{o}
\right\vert
\forall o\in \obj{\bfC}_{\cong}: \psi_o\in \bR_{\geq0}
\land\!\!\!\!\!\!
\sum_{o\in \obj{\bfC}_{\cong}}\!\!\!\!\!\!\psi_o\leq 1
\right\}
\end{equation}	
Let $Stoch(\bfC):=End_{\bR}(Prob(\bfC))$ be the space of endomorphisms of $Prob(\bfC)$, with elements referred to as \emph{sub-stochastic operators}. Then a \textbf{\emph{continuous-time Markov chain (CTMC)}} is specified in terms of a tuple of data $(\ket{\Psi(0)},H)$, where $\ket{\Psi(0)}\in Prob(\bfC)$ is the \emph{initial state}, and where $H\in End_{\bR}(\cS_{\bfC})$ is the \emph{infinitesimal generator} or \emph{Hamiltonian} of the CTMC (with $\cS_{\bfC}$ the space of real-valued sequences indexed by elements of $\obj{\bfC}_{\cong}$ and with finite coefficients). $H$ is required to be an infinitesimal (sub-) stochastic operator, which entails that for $H\equiv (h_{o,o'})_{o,o'\in \obj{\bfC}_{\cong}}$ and for all $o,o'\in \obj{\bfC}_{\cong}$,
\begin{equation}\label{def:Hprops}
(i)\; h_{o,o}\leq 0\,,\; (ii) \forall o\neq o':\; h_{o,o'}\geq 0\,,\; 
(iii)\; \sum_{o'} h_{o,o'}=0\,.
\end{equation}
Then this data encodes the \emph{evolution semi-group} $\cE:\bR_{\geq 0}\rightarrow Stoch(\bfC)$ as the (point-wise minimal non-negative) solution of the \emph{Kolmogorov backwards} or \emph{master equation}:
\begin{equation}
\tfrac{d}{dt}\cE(t)=H\cE(t)\,,\; \cE(0)=\mathbb{1}_{Stoch(\bfC)}
\Rightarrow \;\forall t,t'\in \bR_{\geq 0}: \cE(t)\cE(t')=\cE(t+t')
\end{equation}
Consequently, the \emph{time-dependent state} $\ket{\Psi(t)}$ of the system is given by 
\begin{equation}
\forall t\in \bR_{\geq 0}:\quad \ket{\Psi(t)}=\cE(t)\ket{\Psi(0)}\,.
\end{equation}
\end{definition}

An important technical aspect of the above definition of CTMCs is the definition of the relevant space of (sub-)probability distributions in interaction with the definition of the infinitesimal generator $H$ and of the space $\cS_{\bfC}$. Some remarks on this interaction and a short explanation of the relevant mathematical concepts are provided in Appendix~\ref{app:StochMechProof}.

Our main approach in studying CTMCs based on rewriting systems will consist in analyzing the dynamical statistical behavior of so-called observables:
\begin{definition}[Observables; \cite{bp2018}, Def.~7.1]\label{def:obs}
Let $\cO_{\bfC}\subset End_{\bR}(S_{\bfC})$ denote the space of \emph{observables}, defined as the space of \emph{diagonal operators}\footnote{Depending on the concrete case, the eigenvalue $\omega_O(X)$ in $O\ket{X}=\omega_O(X)\ket{X}$ may e.g.\ coincide with the number of occurrences of a pattern in the object $X$ (see also Appendix~\ref{sec:appProofSMF}).},
\begin{equation}
\cO_{\bfC}:=\{O\in End_{\bR}(S_{\bfC})\mid \forall X\in \obj{\bfC}_{\cong}:\; O\ket{X}=\omega_O(X)\ket{X}\,,\; \omega_O(X)\in \bR\}\,.
\end{equation}
We furthermore define the so-called \emph{projection operation} $\bra{}:S_{\bfC}\rightarrow \bR$ via extending by linearity the definition of $\bra{}$ acting on basis vectors of $\hat{\bfC}$,
\begin{equation}
\forall X\in \obj{\bfC}_{\cong}:\quad \braket{}{X}:=1_{\bR}\,.
\end{equation}
These definitions induce a notion of \emph{correlators} of observables (also referred to as (mixed) moments), defined for $O_1,\dotsc,O_n\in \cO_{\bfC}$ and $\ket{\Psi}\in Prob(\bfC)$ as
\begin{equation}
\langle O_1,\dotsc,O_n\rangle_{\ket{\Psi}}:=\bra{}O_1,\dotsc,O_n\ket{\Psi}
=\sum_{X\in \obj{\bfC}_{\cong}}\psi_X\cdot\omega_{O_1}(X)\cdots \omega_{O_n}(X)\,.
\end{equation}
\end{definition}

The precise relationship between the notions of CTMCs and SqPO-type rewriting rules as encoded in the corresponding SqPO-type rule algebra formalism is established in the form of the following theorem, where in particular the notion of observables is quite different in nature to the DPO-type analogon (compare Thm.~7.12 of~\cite{bp2018}). This result is the first-of-its-kind \emph{universal} definition of SqPO-type stochastic rewriting systems with ``mass-action semantics'' (where activities of productions are proportional to their number of admissible matches in a given system state).

\begin{restatable}[SqPO-type stochastic mechanics framework]{theorem}{thmStochMechSqPO}\label{thm:smfSqPO}
Let $\bfC$ be a category satisfying Assumption~\ref{ass:RAsqpo}. Let $\{(\grule{O_j}{p_j}{I_j})\in \cR^{sq}_{\bfC}\}_{j\in \cJ}$ be a (finite) set of rule algebra elements, and $\{\kappa_j\in \bR_{\geq 0}\}_{j\in \cJ}$ a collection of non-zero parameters (called \emph{base rates}). Then one may construct the Hamiltonian $H$ of the associated CTMC from this data according to
\begin{equation}
H:=\hat{H}+\bar{H}\,,\quad 
\hat{H}:=\sum_{j\in \cJ}\kappa_j\cdot \rho_{\bfC}^{sq}\left(\grule{O_j}{p_j}{I_j}\right)\,,\quad 
\bar{H}:=-\sum_{j\in \cJ}\kappa_j\cdot \bO_{I_j}^{sq}\,.
\end{equation}
Here, the notation $\bO_M^{sq}$ for arbitrary objects $M\in \obj{\bfC}$ denotes the \emph{observables} (sometimes referred to as \emph{motif counting observables}) for the resulting CTMC of SqPO-type, with
\begin{equation}
\bO_M^{sq}:=\rho_{\bfC}^{sq}\left(\delta\left(
M\xleftarrow{id_M}M\xrightarrow{id_M}M
\right)\right)\,.
\end{equation}
We furthermore have the \emph{SqPO-type jump-closure property}, whereby for all $(\grule{O}{p}{I})\in \cR^{sq}_{\bfC}$
\begin{equation}\label{eq:ojcSqPO}
\bra{}\rho^{sq}_{\bfC}(\grule{O}{p}{I})=\bra{}\bO_I^{sq}\,.
\end{equation}
\end{restatable}
\begin{proof}
See Appendix~\ref{app:StochMechProof}.
\end{proof}

\section{Application example: a dynamical random graph model}
\label{sec:appEx}

In order to illustrate our novel SqPO-type stochastic mechanics framework, let us consider a dynamical system evolving on the space of finite directed multigraphs.

\begin{example}
	Let $\mathbf{FinGraph}$ be the finitary restriction of the category $\mathbf{Graph}$ (see also Lemma~\ref{lem:GraphFPC}), and denote by $\inObj\in \mathbf{FinGraph}$ the strict initial object (the empty graph). We define a stochastic SqPO rewriting system based upon rules encoding \emph{vertex creation/deletion} ($v_{\pm}$) and \emph{edge creation/deletion} ($e_{\pm}$):
	\begin{equation}
	\begin{aligned}
		v_{+}&:=(\OneVertG\leftarrow \inObj\rightarrow \inObj)
		&\qquad  
		v_{-}&:=(\inObj\leftarrow \inObj\rightarrow \OneVertG)\\
		e_{+}&:=(\TwoVertDirEdgeG[] \leftarrow \TwoVertG[]\rightarrow \TwoVertG[])
		&\qquad  
		e_{-}&:=(\TwoVertG[] \leftarrow \TwoVertG[]\rightarrow \TwoVertDirEdgeG[])
	\end{aligned}
	\end{equation}
	Together with a choice of \emph{base rates} $\nu_{\pm},\varepsilon_{\pm}\in \bR_{\geq0}$ and an initial state $\ket{\Psi(0)}\in Prob(\mathbf{FinGraph})$, this data defines a stochastic rewriting system with Hamiltonian $H:=\hat{H}+\bar{H}$,
	\begin{equation}
	\begin{aligned}
		\hat{H}&=\nu_{+}V_{+}+\nu_{-}V_{-}+\varepsilon_{+}E_{+}+\varepsilon_{-}E_{-}\\
		\bar{H}&=-\nu_{+}\bO_{\inObj}-\nu_{-}\bO_{\OneVertG}
		-\varepsilon_{+}\bO_{\OneVertG\; \OneVertG}-\varepsilon_{-}\bO_{\TwoVertDirEdgeG[]}\,,
	\end{aligned}
	\end{equation}
	where $V_{\pm}:=\rho^{sq}_{\mathbf{FinGraph}}(\delta(v_{\pm}))$ and $E_{\pm}:=\rho^{sq}_{\mathbf{FinGraph}}(\delta(e_{\pm}))$.
\end{example}

Despite the apparent simplicity of this model (which might be seen as a paradigmatic example of a \emph{random graph model}), the explicit analysis via the stochastic mechanics framework will uncover a highly non-trivial interaction of the dynamics of the vertex- and of the edge-counting observables. Intuitively, since in SqPO-rewriting no conditions are posed upon vertices that are to be deleted, the model is expected to possess a vertex dynamics that is the one of a so-called \emph{birth-death process}. If it were not for the vertex deletions, one would find a similar dynamics for the edge-counting observables (compare e.g.\ the DPO-type rewriting model considered in~\cite{bp2018}). However, since deletion of vertices deletes all incident edges, the dynamics of the edge-counting observable is rendered considerably more complicated, and in particular much less evident to foresee by heuristic arguments.\\

In order to compute the dynamics of the vertex counting observable $O_V:=\bO_{\OneVertG}$, we follow the approach of \emph{exponential moment-generating functions} put forward in~\cite{bdg2016,bdp2017,bdg2019} and define
\begin{equation}
	M_V(t;\lambda):=\bra{}e^{\lambda O_V}\ket{\Psi(t)}\,,
\end{equation}
with $\lambda$ a formal variable. $M_V(t;\lambda)$ encodes the moments of the observable $O_V$, in that taking the $n$-th derivative of $M_V(t;\lambda)$ w.r.t.\ $\lambda$ followed by setting $\lambda\to0$ yields the $n$-th moment of $O_V$. Note that we must assume the \emph{finiteness} of all statistical moments as standard in the probability theory literature in order for $M_V(t;\lambda)$ to be well-posed, a property that we will in the case at hand indeed derive explicitly. Referring the interested readers to~\cite{bp2019-ext} for further details, suffice it here to recall the following variant of the BCH formula  (see e.g.\  \cite{hall2015lieGroups}, Prop.~3.35), for $\lambda$ a formal variable and $A,B$ two composable linear operators,
\begin{equation}
e^{\lambda A}Be^{-\lambda A}=e^{ad_{\lambda A}}B
	=\sum_{n\geq 0}\frac{\lambda^n}{n!} ad_{ A}^{\circ n}(B)\,,\quad
ad_A(B):=AB-BA\equiv [A,B]\,,
\end{equation}
with the convention that $ad_A^{\circ 0}(B):=B$. The operation $[.,.]$ is typically referred to as the \emph{commutator}. We may then derive the \emph{formal evolution equation} for $M_V(t;\lambda)$:
\begin{equation}
\begin{aligned}
	\tfrac{\partial}{\partial t}M_V(t;\lambda)&=\bra{}e^{\lambda O_V}H\ket{\Psi(t)}=\bra{}\left(e^{\lambda O_V}H e^{-\lambda O_V}\right)e^{\lambda O_V}\ket{\Psi(t)}\\
	&=\bra{}\left(e^{ad_{\lambda O_V}}H\right)e^{\lambda O_V}\ket{\Psi(t)}\,.
\end{aligned}
\end{equation}
Since by definition $\bra{}H=0$, it remains to compute the adjoint action $ad_{O_V}(H)$ of $O_V$ on $H$:
\begin{equation}
\begin{aligned}
	ad_{O_V}(H)&=\nu_{+}[O_V,V_{+}]+\nu_{-}[O_V,V_{-}]+\varepsilon_{+}[O_V,E_{+}]+\varepsilon_{-}[O_V,E_{-}]\\
	&=\nu_{+}V_{+}-\nu_{-}V_{-}
\end{aligned}
\end{equation}
Here, the result that $[O_V,E_{\pm}]=0$ has a very simple intuitive meaning: in applications of the linear rules $e_{\pm}$, the number of vertices remains unchanged, whence the vanishing of the commutator. Combining these results with the SqPO-type \emph{jump-closure property} (cf.\ Theorem~\ref{thm:smfSqPO}), we finally arrive at the following \emph{formal evolution equation} for $M_V(t;\lambda)$:
\begin{equation}\label{eq:FEQ}
\begin{aligned}
\tfrac{\partial}{\partial t}M_V(t;\lambda)&=
	\nu_{+}\left(e^{\lambda}-1\right)\bra{}V_{+}e^{\lambda O_V}\ket{\Psi(t)}
	+\nu_{-}\left(e^{-\lambda}-1\right)\bra{}V_{-}e^{\lambda O_V}\ket{\Psi(t)}\\
	&\overset{\eqref{eq:ojcSqPO}}{=}
	\nu_{+}\left(e^{\lambda}-1\right)\bra{}e^{\lambda O_V}\ket{\Psi(t)}
	+\nu_{-}\left(e^{-\lambda}-1\right)\bra{}O_Ve^{\lambda O_V}\ket{\Psi(t)}\\
	&=\left(
	\nu_{+}\left(e^{\lambda}-1\right)+\nu_{-}\left(e^{-\lambda}-1\right)\tfrac{\partial}{\partial \lambda}
	\right)M_V(t;\lambda)\,.
\end{aligned}
\end{equation}
Supposing for simplicity an initial state $\ket{\Psi(0)}=\ket{G_0}$ (for $G_0\in \obj{\mathbf{Graph}_{fin}}$ some graph with $N_V$ vertices and $N_E$ edges), we  find that $M_V(0;\lambda)=\exp(\lambda N_V)$. The resulting initial value problem may be solved in closed-form via \emph{semi-linear normal-ordering} techniques known from the combinatorics literature~\cite{Dattoli:1997iz,blasiak2005boson,blasiak2011combinatorial,bdp2017} (see also~\cite{bp2019-ext,bdg2019}), and we obtain (for $t\geq 0$)
\begin{equation}\label{eq:MVsol}
M_V(t;\lambda)=\exp\left({\frac{\nu_{+}}{\nu_{-}}(e^{\lambda}-1)(1-e^{-\nu_{-}t})}\right)\left(1+(e^{\lambda}-1)e^{-\nu_{-}t}
\right)^{N_V}\,.
\end{equation}
In the limit $t\to\infty$, the moment-generating function becomes that of a \emph{Poisson-distribution} (of parameter $\nu_{+}/\nu_{-}$), thus confirming the aforementioned intuition that the vertex-counting observable has the dynamical behavior of a so-called \emph{birth-death process} (see e.g.\ \cite{bdp2017}).\\

Let us consider next the dynamics of the edge-counting observable $O_E:=\bO_{\TwoVertDirEdgeG[]}$, where for brevity we will only consider the evolution of the mean edge count. The calculation of the evolution equation for the expectation value of $O_E$ simplifies to the analogue of the so-called \emph{Ehrenfest equation},
\begin{equation}
\begin{aligned}
\tfrac{\partial}{\partial t}\bra{}O_E\ket{\Psi(t)}&=
\bra{}O_E \,H\ket{\Psi(t)}=\bra{}\big(H\, O_E+[O_E,H]\big)\ket{\Psi(t)}\,.
\end{aligned}
\end{equation}
Recalling that $\bra{}H=0$, it remains to compute the commutator $[O_E,H]$:
\begin{equation}
\begin{aligned}
[O_E,H]&=\nu_{+}[O_E,V_{+}]+\nu_{-}[O_E,V_{-}]+\varepsilon_{+}[O_E,E_{+}]+\varepsilon_{-}[O_E,V_{-}]\\
&=\nu_{+}\cdot 0 -\nu_{-}(E_{-}^{0,1}+E_{-}^{1,0})+\varepsilon_{+}E_{+}-\varepsilon_{-}E_{-}\\
E_{-}^{0,1}&=\rho_{\mathbf{FinGraph}}^{sq}\left(\delta\left(\OneVertGL[]{b}\leftarrow\OneVertGL[]{b}\rightarrow \TwoVertDirEdgeGLb{}{a}{b}\right)\right)\\
E_{-}^{1,0}&=\rho_{\mathbf{FinGraph}}^{sq}\left(\delta\left(\OneVertGL[]{a}\leftarrow\OneVertGL[]{a}\rightarrow \TwoVertDirEdgeGLb{}{a}{b}\right)\right)\,.
\end{aligned}
\end{equation}
This calculation is a representative example of various effects that may occur in rule-algebraic commutation relations: we find a zero commutator $[O_E,V_{+}]$, indicating the fact that application of the vertex creation rule $V_{+}$ does not influence the edge count. The commutators $[O_E,E_{\pm}]=\pm E_{\pm}$ encode that application of the edge creation/deletion rules leads to positive/negative contributions to the edge count. Finally, the contribution of the commutator $[O_E,V_{-}]=-E_{-}^{0,1}-E_{-}^{1,0}$ is given by the representations of two rule algebra elements not originally present in the Hamiltonian $H$, with the structure of the underlying linear rules indicated by the labels $a$ and $b$ on the vertices (as customary in the rewriting literature). It then remains to apply the jump-closure property (Theorem~\ref{thm:smfSqPO}) together with the  identity $\bO_{\OneVertG\; \OneVertG}=O_V(O_V-1)$ in order to obtain the \emph{evolution equation}
\begin{equation}\label{eq:evoEmean}
\tfrac{\partial}{\partial t}\bra{}O_E\ket{\Psi(t)}
=\varepsilon_{+}\bra{}O_V(O_V-1)\ket{\Psi(t)}-(\varepsilon_{-}+2\nu_{-})\bra{}O_E\ket{\Psi(t)}\,.
\end{equation}
Together with an initial condition such as e.g.\ $\ket{\Psi(0)}=\ket{G_0}$ for some (finite) directed graph $G_0$ with $N_V$ vertices and $N_E$ edges, and computing the closed-form expression for the first contribution in~\eqref{eq:evoEmean} from our previous solution~\eqref{eq:MVsol} (as $\partial_{\lambda}(\partial_{\lambda}-1)M_V(t;\lambda)$ followed by setting $\lambda\to0$), the initial value problem for the mean edge count evolution may be easily solved in closed form via the use of a computer algebra software such as \textsc{Maple}, \textsc{Mathematica} or \textsc{Sage}. It is also straightforward to verify that for an arbitrary initial state $\ket{\Psi(0)}=\ket{G_0}$, the limit value of the mean edge count for $t\to\infty$ reads
\begin{equation}
	\lim\limits_{t\to \infty}\bra{}O_E\ket{\Psi(t)}=\tfrac{\nu_{+}^2\varepsilon_{+}}{\nu_{-}^2(2\nu_{-}+\varepsilon_{-})}\,.
\end{equation}
Since the rates $\nu_{\pm}$ and $\varepsilon_{\pm}$ are free parameters, the above result entails that in this model one may freely adjust the limit value of the average vertex count as encoded in~\eqref{eq:FEQ} (whence $\nu_{+}/\nu_{-}$) as well as the limit value of the average edge count via suitable choices of the parameters $\varepsilon_{\pm}$. For illustration, we present some plots of the mean edge count evolution for the case $\ket{\Psi(0)}=\ket{\inObj}$ and various choices of parameters in Figure~\ref{fig:meanEdgeEvo}. 

\begin{figure}
  \centering
    \includegraphics[width=0.9\textwidth]{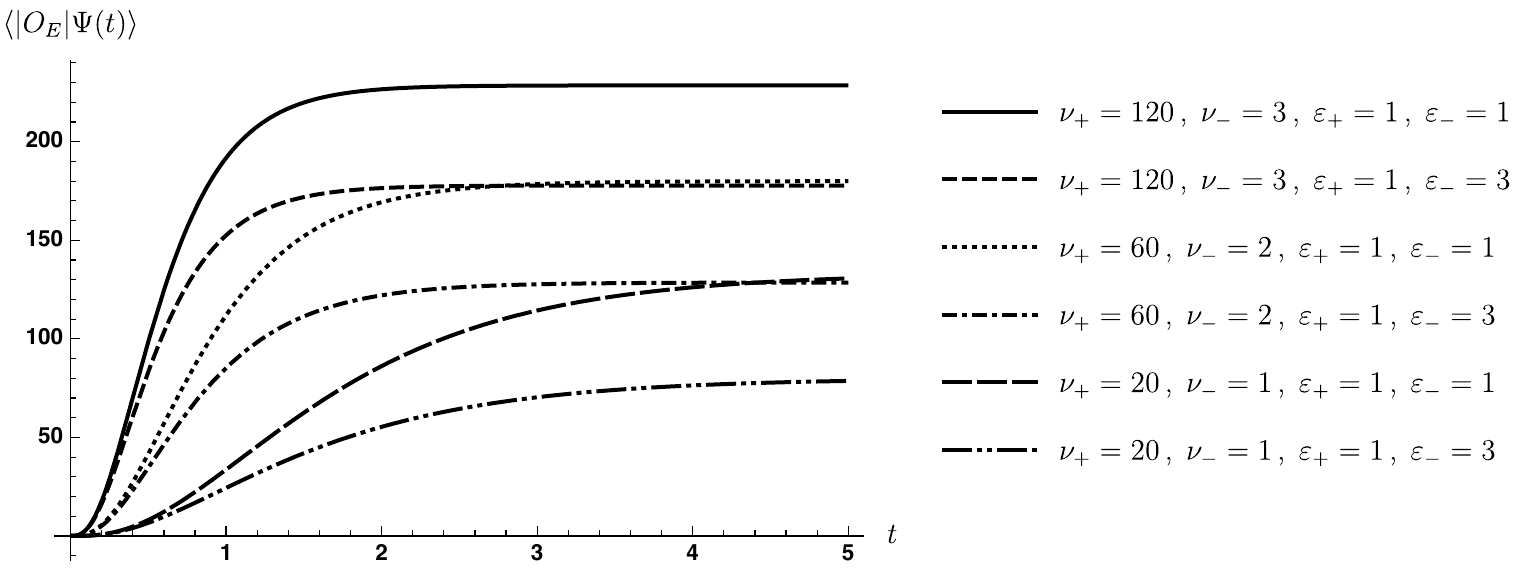}
     \caption{Time-evolution of $\bra{}O_E\ket{\Psi(t)}$ for $\ket{\Psi(0)}=\ket{\inObj}$.\label{fig:meanEdgeEvo}}
\end{figure}

\section{Conclusion and Outlook}

Extending our previous work on Double-Pushout (DPO) rewriting theories as presented in~\cite{bdg2016,bp2018,bdg2019} to the important alternative setting of Sesqui-Pushout (SqPO) rewriting, we provide a number of original results in the form of \emph{concurrency} and \emph{associativity} theorems for SqPO rewriting theories on adhesive categories. %
These fundamental results in turn permit us to formulate so-called \emph{SqPO-type rule algebras}, which play a central role in our novel \emph{universal stochastic mechanics framework}. %
We strongly believe that these contributions will provide fruitful grounds for further developments both in theory and practice of rewriting beyond the specialists' communities, especially in view of static analysis techniques~\cite{b2019c}.

%
%
%
\bibliographystyle{eptcs}
\bibliography{refs-NB-GCM-2019-proceedings.bib}

\appendix

\section{A collection of useful technical results on adhesive categories and final pullback complements}\label{app:lemList}

\textbf{Notational convention:} Here and as throughout this paper, while evidently category-theoretical constructions such as pushouts are only unique up to isomorphisms, we will typically nevertheless pick convenient representatives of the respective isomorphism classes to simplify our notations. As standard practice in the literature, we will thus e.g.\ fix the convention as in~\eqref{eq:lemSsqEqs} to choose representatives appropriately to label the pushout along an isomorphism with ``equality arrows'' (rather than keeping object labels generic and decorating the relevant arrow with a ``$\cong$'' symbol).

\begin{lemma}\label{lem:Main}
Let $\bfC$ be a category.
\begin{enumerate}
\item \emph{``Single-square'' lemmata~(see e.g.\ \cite{bp2019-ext}, Lem.~1.7):}
In any category, given commutative diagrams of the form
 \begin{equation}\label{eq:lemSsqEqs}
    \begin{mycd}
      A \ar[r,"f"] \ar[d,equal]\ar[dr,phantom,"(A)"] & B \ar[d,equal]\\
      A \ar[r,"f"'] & B
    \end{mycd}\qquad \begin{mycd}
      A \ar[r,equal] 
      \ar[d,equal]\ar[dr,phantom,"(B)"] & A \ar[d,hook,"g"]\\
      A \ar[r,hook,"g"'] & B
    \end{mycd}\qquad \begin{mycd}
      A \ar[r,"f"] 
      \ar[d,equal]\ar[dr,phantom,"(C)"] & 
      B \ar[d,hook,"g"]\\
      A \ar[r,"g\circ f"'] & C
    \end{mycd}\,,
  \end{equation}
\begin{enumerate}
\item $(A)$ is a pushout for arbitrary morphisms $f$\label{lem:fPOPB}, 
\item $(B)$ is a pullback if and only if the morphism $g$ is a monomorphism\label{lem:monoPB}, and
\item $(C)$ is a pullback for arbitrary morphisms $f$ if $g$ is a monomorphism\label{lem:idPB}.
\end{enumerate}
\item \emph{special adhesivity corollaries} (cf.\ e.g.\ \cite{EHRIG:2014ma}, Lemma~2.6): in any adhesive category,
\begin{enumerate}
	\item pushouts along monomorphisms are also pullbacks, and
	\item (\emph{uniqueness of pushout complements}) given a monomorphism $A\hookrightarrow C$ and a generic morphism $C\rightarrow D$, the respective pushout complement $A\rightarrow B \xhookrightarrow{b} D$ (if it exists) is unique up to isomorphism, and with $b\in \mono{\bfC}$ (due to stability of monomorphisms under pushouts).
\end{enumerate}
\item \emph{``Double-square lemmata''}: given commutative diagrams of the shapes
\begin{equation}
\vcenter{\hbox{\begin{mycd}
A
	\ar[d,"a"']\ar[dr,phantom,"(1)"] & 
B 
	\ar[l,"d"']\ar[d,"b" description]\ar[dr,phantom,"(2)"]  &
C
	\ar[l,"e"']\ar[d,"c"]\\
A' & 
B'\ar[l,"d'"] & C'\ar[l,"e'"]
\end{mycd}}}\qquad
\vcenter{\hbox{\begin{mycd}
Z
	\ar[d,"w"']\ar[dr,phantom,"(3)"] & 
Z' 
	\ar[d,"w'"]
	\ar[l,"z"']\\
Y \ar[d,"v"']\ar[dr,phantom,"(4)"] &
Y'	\ar[l,"y" description]
	\ar[d,"v'"]\\
X & 
X'\ar[l,"x"]
\end{mycd}}}
\end{equation}
then in any category $\bfC$ (cf.\ e.g.\ \cite{lack2005adhesive}):
\begin{enumerate}
\item \emph{Pullback-pullback (de-)composition}: If $(1)$ is a pullback, then $(1)+(2)$ is a pullback if and only if $(2)$ is a pullback.\label{lem:PBPBdec}
\item \emph{Pushout-pushout (de-)composition}: If $(2)$ is a pushout, then $(1)+(2)$ is a pushout if and only if $(1)$ is a pushout.\label{lem:POPOdec}
\end{enumerate}
If the category is adhesive: 
\begin{enumerate}\setcounter{enumii}{2}
\item \emph{pushout-pullback decomposition} (\cite{EHRIG:2014ma}, Lemma~2.6): If $(1)+(2)$ is a pushout, $(1)$ is a pullback, and if $d'\in \mono{\bfC}$ and ($c\in \mono{\bfC}$ or $e\in \mono{\bfC}$), then $(1)$ and $(2)$ are both pushouts (and thus also pullbacks).\label{lem:POPBdec}
\item \emph{pullback-pushout decomposition} (\cite{GOLAS2014}, Lem.~B.2): if $(1)+(2)$ is a pullback, $(2)$ a pushout, $(1)$ commutes and $a\in \mono{\bfC}$, then $(1)$ is a pullback.\label{lem:PBPOdec}
\item \emph{Horizontal FPC (de-)composition} (cf.\ \cite{Corradini_2006}, Lem.~2 and Lem.~3, compare~\cite{Loewe_2015}, Prop.~36):\footnote{It is worthwhile emphasizing that in these FPC-related lemmata, the ``orientation'' of the diagrams plays an important role. Moreover, the precise identity of the pair of morphisms that plays the role of the final pullback complement in a given square may be inferred from the ``orientation'' specified in the condition part of each statement.} If $(1)$ is an FPC (i.e.\ if $(d',b)$ is FPC of $(a,d)$), then $(1)+(2)$ is an FPC if and only if $(2)$ is an FPC.\label{lem:horFPCdec}
\item \emph{Vertical FPC (de-)composition} (ibid):\label{lem:vertFPCdec} if $(3)$ is an FPC (i.e.\ if $(y.w')$ is FPC of $(w,z)$), then 
\begin{enumerate}
\item if $(4)$ is an FPC (i.e.\ if $(x,v')$ is FPC of $(v,y)$), then $(3)+(4)$ is an FPC (i.e.\ $(x,v'\circ w')$ is FPC of $(v\circ w,z)$);
\item if $(3)+(4)$ is an FPC (i.e.\ if $(x,v'\circ w')$ is FPC of $(v\circ w,z)$), and if $(4)$ is a pullback, then $(4)$ is an FPC (i.e.\ $(x,v')$ is FPC of $(v,y)$).
\end{enumerate}
\item \emph{Vertical FPC-pullback decomposition} (compare~\cite{Loewe_2015}, Lem.~38): If $v\in \mono{\bfC}$, if $(4)$ is a pullback and if $(3)+(4)$ is an FPC (i.e.\ if $(x,v'\circ w')$ is FPC of $(v\circ w,z)$), then $(3)$ and $(4)$ are FPCs.\label{lem:vertFPCpbDec}
\item \emph{Vertical FPC-pushout decomposition}\footnote{We invite the interested readers to compare the precise formulation of the vertical FPC-pushout decomposition result to its concrete applications in the setting of the proof of the concurrency theorem, for which it has been specifically developed.}: If all morphisms of the squares $(3)$ and $(4)$ except $v$ are in $\mono{\bfC}$, if $v\circ w\in \mono{\bfC}$, if $(3)$ is a pushout and if $(3)+(4)$ is an FPC (i.e.\ if $(x,v'\circ w')$ is FPC of $(v\circ w,z)$), then $(4)$ is an FPC and $v\in \mono{\bfC}$.\label{lem:vertFPCpoDec}
\end{enumerate}
\end{enumerate}
\begin{proof}
Referring to the references above for the proofs of the (well-known) statements (where necessary by specializing the more general case of $\cM$-adhesive categories to the case of adhesive categories via setting $\cM$ to the class of all monomorphisms), it remains to prove our novel vertical FPC-pushout decomposition result. To this end, we first invoke pullback-pushout decomposition (Lemma \ref{lem:Main}\eqref{lem:PBPOdec}) in order to demonstrate that since $(3)+(4)$ is an FPC and thus also a pullback, and since $(3)$ is a pushout and since $x\in \mono{\bfC}$, $(4)$ is a pullback. By applying vertical FPC-pullback decomposition, we may conclude that $(4)$ is an FPC. In order to demonstrate that $v\in \mono{\bfC}$, construct the commutative cube below left:
\begin{equation}
\vcenter{\hbox{\includegraphics[scale=0.8]{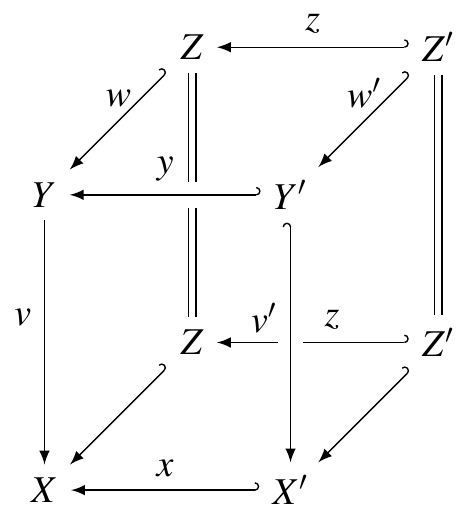}}}\qquad \qquad
\begin{mycd}\gdef\mycdScale{1}
& Z\ar[d,hook,"w"']\ar[ddl,bend right,hook',"v\circ w"'] & 
Z'\ar[dl,phantom,"(3)"]\ar[l,hook',"z"']\ar[d,hook,"w'"]\\
& Y\ar[dl,"v"'] & 
Y'\ar[l,hook',"y"]\ar[dll,bend left,hook,"v\circ y"]\\
X & &\\
\end{mycd}
\end{equation}
Since the bottom square is the FPC (and thus pullback) $(3)+(4)$, and since the right square is a pullback via Lemma~\ref{lem:Main}\eqref{lem:idPB} (because $v'\in \mono{\bfC}$), by pullback composition the square $\cSquare{Z',Z,X,Y'}$ (the right plus the bottom square) is a pullback. Thus assembling the commutative diagram as shown above right, since by assumption $(3)$ is a pushout and all arrows except $v$ are monomorphisms, invoking Theorem~\ref{thm:effUn} below permits to prove that also $v\in \mono{\bfC}$.
\end{proof}
\end{lemma}

The last part of the preceding proof relied upon one of the quintessential properties of adhesive categories in view of associative rewriting theories:
\begin{theorem}[Effective unions; \cite{lack2005adhesive}, Thm.~5.1]\label{thm:effUn}
In an adhesive category $\bfC$, given a commutative diagram such as  in the middle of~\eqref{eq:diags}, if all morphisms except the morphism $x$ are monomorphisms, if the square marked $(A)$ is a pushout and if the exterior square is a pullback, then $x$ is also a monomorphism.
\end{theorem}

The following result provides several important facts on FPCs.

\begin{lemma}[cf.\ \cite{Loewe_2015}, Fact~2, and~\cite{Corradini_2006}, Lemma~2 and Proposition~2]\label{lem:FPCfacts}
	Let $\bfC$ be adhesive. For an arbitrary morphism $f:A\rightarrow B$, $(id_B,f)$ is an FPC of $(f,id_A)$ and vice versa. Moreover, every pushout square along monomorphisms is also an FPC square. FPCs are unique up to isomorphism and preserve monomorphisms. The latter property entails that if $C\xleftarrow{d}D\xleftarrow{b}A$ is the FPC of $C\xleftarrow{c}B\xleftarrow{a}A$ and if $a\in \mono{\bfC}$, then also $d\in \mono{\bfC}$ and vice versa (while $c\in \mono{\bfC}$ entails that $b\in \mono{\bfC}$ by stability of monomorphisms under pullbacks in an adhesive category $\bfC$).
\end{lemma}

For concreteness, we quote the following explicit construction of FPCs in the category $\mathbf{Graph}$ of directed multigraphs:

\begin{lemma}[FPCs for graphs; \cite{Corradini_2006}, Sec.~4.1 and Construction~5]\label{lem:GraphFPC}
	Let $\mathbf{Graph}$ denote the adhesive category of \emph{directed multigraphs}, with
	\begin{itemize}
	\item $\obj{\mathbf{Graph}}$: \emph{(multi-)graphs}, i.e.\ tuples $G=(V_G,E_G,\mathsf{src}_G:E_G\rightarrow V_G,\mathsf{trg}_G:E_G\rightarrow V_G)$, with $V_G$ the set of vertices, $E_G$ the set of edges (with $V_G\cap E_G=\emptyset$), $\mathsf{src}_G$ the source and $\mathsf{trg}_G$ the target maps
	\item $\mor{\mathbf{Graph}}$: \emph{graph homomorphisms} $f:G\rightarrow H$, specified in terms of pairs of morphisms $(f_V:V_G\rightarrow V_H,f_E:E_G\rightarrow E_H)$ such that $\mathsf{src}_H\circ f_E=f_V\circ \mathsf{src}_G$ and $\mathsf{trg}_H\circ f_E=f_V\circ \mathsf{trg}_G$.
\end{itemize}
Let $\mono{\mathbf{Graph}}$ denote the class of all injective graph morphisms. Then for every composable pair of monomorphisms $K\xhookrightarrow{i}I\xhookrightarrow{m}X$, the FPC exists and is constructed explicitly as\footnote{The quoted Construction~5 of~\cite{Corradini_2006} is slightly more general, in that the morphism $m$ may be permitted to not be a monomorphism; we will however have no application for such a generalization in our framework.} $K\xhookrightarrow{m\vert K}\overline{K}\xhookrightarrow{\subseteq}X$, where $\xrightarrow{\subseteq}$ denotes an inclusion morphism, and where the graph $\overline{K}$ reads
\begin{equation}
\begin{aligned}
	V_{\overline{K}}&=V_X\setminus m[V_I\setminus V_K]\\
	E_{\overline{K}}&=\{
		e\in E_X\setminus m[E_I\setminus E_K]\vert \mathsf{src}_X(e)\in V_{\overline{K}}\land 
		\mathsf{trg}_X(e)\in V_{\overline{K}}
	\}\,.
\end{aligned}
\end{equation}
\end{lemma}

\section{Proofs}

\subsection{Proof of the SqPO concurrency theorem}\label{app:SqPOconcur}

\begin{proof}
Throughout this proof, in each individual constructive step it may be verified that due to the stability of monomorphisms under pullbacks and pushouts, due to the various decomposition lemmata provided in the form of Lemma~\ref{lem:Main}, and on occasion due to Theorem~\ref{thm:effUn} on effective unions in adhesive categories, all morphisms induced in the ``Synthesis'' and ``Analysis'' steps are in fact monomorphisms. For better readability, we will not explicitly mention the individual reasoning steps on this point except for a few intricate sub-steps, since they may be recovered in a straightforward manner.\\

--- \textbf{Synthesis:} %
	Consider the setting presented in~\eqref{eq:CTs1}. %
	Here, we have obtained the candidate match ${\color{h1color}\mathbf{n}}=(I_2{\color{h1color}\leftarrow M_{21}\rightarrow} O_1)$ via pulling back the cospan $(I_2{\color{h1color}\rightarrow} X_1{\color{h1color}\leftarrow} O_1)$. %
	Next, we construct ${\color{h1color}N_{21}}$ via taking the pushout of ${\color{h1color}\mathbf{n}}$, which induces a unique arrow ${\color{h1color}N_{21}\rightarrow}X_1$ that is according to Theorem~\ref{thm:effUn} a monomorphism. %
	The diagram in~\eqref{eq:CTs2} is obtained by taking the pullbacks of the spans $\overline{K}_i\rightarrow X_1{\color{h1color}\leftarrow N_{21}}$ (obtaining the objects $K_i'$, for $i=1,2$). %
	By virtue of pushout-pullback decomposition (Lemma~\ref{lem:Main}\eqref{lem:POPBdec}), the squares $\cSquare{K_1',\overline{K}_1,X_1,{\color{h1color}N_{21}}}$ and $\cSquare{K_1,K_1',{\color{h1color}N_{21}},O_1}$ are pushouts. %
	Invoking vertical FPC-pullback decomposition (Lemma~\ref{lem:Main}\eqref{lem:vertFPCpbDec}), the squares $\cSquare{K_2',{\color{h1color}N_{21}},X_1,\overline{K}_2}$ and $\cSquare{K_2,I_2,{\color{h1color}N_{21}},K_2'}$ are FPCs. %
	Next, letting ${\color{h2color}O_{21}}:=\pO{O_2\leftarrow K_2\rightarrow K_2'}$ and ${\color{h2color}I_{21}}:=\pO{O_1\leftarrow K_1\rightarrow K_1'}$, we have via vertical FPC-pushout decomposition (Lemma~\ref{lem:Main}\eqref{lem:vertFPCpoDec}) that the resulting two squares on the very right (the ones involving ${\color{h2color}I_{21}}$) are FPCs and that the arrow ${\color{h2color}I_{21}\rightarrow}X_0$ is a monomorphism, while pushout-pushout decomposition (Lemma~\ref{lem:Main}\eqref{lem:POPOdec}) entails that the two newly formed squares on the very left (the ones involving ${\color{h2color}O_{21}}$) are pushouts. %

	The final step as depicted in~\eqref{eq:CTs3} consists in constructing ${\color{h2color}K_{21}}=\pB{K_2'\rightarrow {\color{h1color}N_{21}}\leftarrow K_1'}$ and ${\color{h2color}\overline{K}_{21}}=\pB{\overline{K_2}\rightarrow X_1\leftarrow \overline{K_1}}$, which by universality of pullbacks induces a unique arrow ${\color{h2color}K_{21}\rightarrow \overline{K}_{21}}$. %
	By invoking pullback-pullback decomposition (Lemma~\ref{lem:Main}\eqref{lem:PBPBdec}), one may show that the squares $\cSquare{{\color{h2color}K_{21}},{\color{h2color}\overline{K}_{21}},\overline{K}_i,K_i'}$ (for $i=1,2$) are pullbacks. %
	Since the square $\cSquare{K_1',\overline{K}_1,X_1,{\color{h1color}N_{21}}}$ is a pushout, via the van Kampen property (cf.\ Def.~\ref{def:adhCats}) the square $\cSquare{{\color{h2color}K_{21}},{\color{h2color}\overline{K}_{21}},\overline{K}_2,K_2'}$ is a pushout. %
	Since according to Lemma~\ref{lem:FPCfacts} pushouts are also FPCs, it follows via horizontal composition of FPCs (Lemma~\ref{lem:Main}\eqref{lem:horFPCdec}) that the square $\cSquare{{\color{h2color}K_{21}},{\color{h2color}\overline{K}_{21}},X_1,{\color{h1color}N_{21}}}$ is an FPC. %
	Noting that the pushout square $\cSquare{K_1',\overline{K}_1,X_1,{\color{h1color}N_{21}}}$ is an FPC as well, it follows via horizontal decomposition of FPCs (Lemma~\ref{lem:Main}\eqref{lem:horFPCdec}) that $\cSquare{{\color{h2color}K_{21}},{\color{h2color}\overline{K}_{21}},\overline{K}_1,K_1'}$ is an FPC. %
	Thus the claim follows by invoking pushout composition (Lemma~\ref{lem:Main}\eqref{lem:POPOdec}) and horizontal FPC composition (Lemma~\ref{lem:Main}\eqref{lem:horFPCdec}) in order to obtain the pushout square $\cSquare{{\color{h1color}K_{21}},{\color{h1color}\overline{K}_{21}},X_2,{\color{h1color}O_{21}}}$ and the FPC square $\cSquare{{\color{h1color}K_{21}},{\color{h1color}\overline{K}_{21}},X_0,{\color{h1color}I_{21}}}$.

	--- \textbf{Analysis:} Given the setting as depicted in~\eqref{eq:CTa1} of Figure~\ref{fig:CTa}, where the top row has the structure of an SqPO-composition (compare~\eqref{eq:SqPOccomp}), where the square $\cSquare{{\color{h2color}K_{21}},K_1',{\color{h1color}N_{21}},K_2'}$ is a pullback, the left ``curvy'' bottom square a pushout and the right ``curvy'' bottom square an FPC, we may obtain the configuration of~\eqref{eq:CTa3} as follows: %
	construct\footnote{Note that it is precisely in this step and the following step that we require the existence of FPCs for arbitrary pairs of monomorphisms as per Assumption~\ref{ass:SqPO}.} $\overline{K}_1$ via taking the final pullback complement of $K_1'\rightarrow I_{21}\rightarrow X_0$ (which implies the existence of an arrow ${\color{h2color}\overline{K}_{21}}\rightarrow \overline{K}_1$ via the FPC property). %
	Note in particular that according to Lemma~\ref{lem:FPCfacts}, both arrows constructed via forming the aforementioned FPC are monomorphisms, and thus by stability of monomorphisms in an adhesive category (compare Definition~\ref{def:adhCats} and Lemma~\ref{lem:Main}\eqref{lem:idPB}), the arrow ${\color{h2color}\overline{K}_{21}}\rightarrow \overline{K}_1$ is a monomorphism as well. %
	Next, take the pushout $X_1=\pO{\overline{K}_1\leftarrow K_1'\rightarrow {\color{h1color}N_{21}}}$, followed by constructing $\overline{K}_2$ as the final pullback complement of $K_2'\rightarrow {\color{h1color}N_{21}}\rightarrow X_1$ (which implies due to the FPC property of the resulting square $\cSquare{K_2',\overline{K}_2,X_1,{\color{h1color}N_{21}}}$ the existence of an arrow ${\color{h2color}\overline{K}_{21}}\rightarrow \overline{K}_2$). %
	Invoking pullback-pullback decomposition (Lemma~\ref{lem:Main}\eqref{lem:PBPBdec}) twice, followed by the van Kampen property (Def.~\ref{def:adhCats}), we may conclude that the square $\cSquare{{\color{h2color}K_{21}},{\color{h2color}\overline{K}_{21}}, \overline{K}_2,K_2'}$ is a pushout. Thus invoking pushout-pushout decomposition (Lemma~\ref{lem:Main}\eqref{lem:POPOdec}), we find that also $\cSquare{K_2',\overline{K}_2,X_2,O_{21}}$ is a pushout. %
	We finally arrive at the configuration in~\eqref{eq:CTa4} via composition of pushout and FPC squares, respectively, thus concluding the proof.
\end{proof}

\begin{figure}[ht!]
\begin{subequations}
\begin{align}
	\vcenter{\hbox{\includegraphics[scale=0.6,page=1]{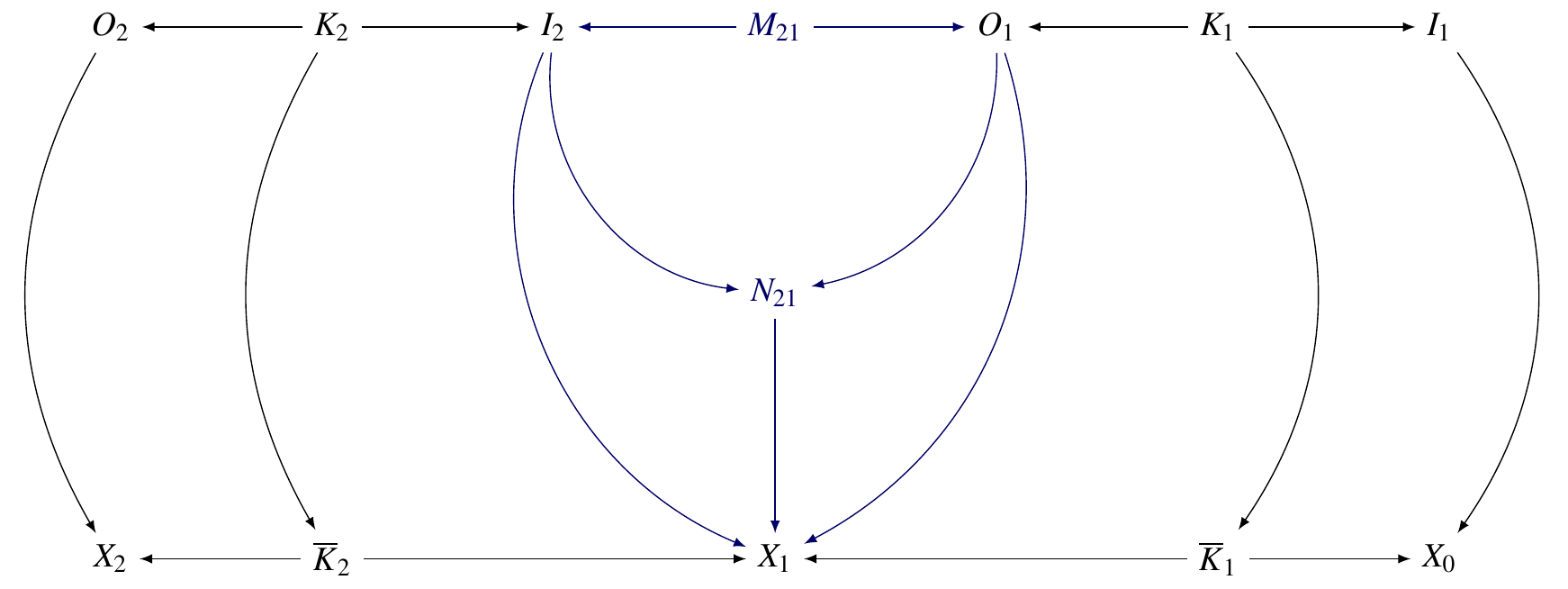}}}\label{eq:CTs1}\\
	\vcenter{\hbox{\includegraphics[scale=0.6,page=2]{images/concurrency-proof.pdf}}}\label{eq:CTs2}\\
	\vcenter{\hbox{\includegraphics[scale=0.6,page=3]{images/concurrency-proof.pdf}}}\label{eq:CTs3}
\end{align}
\end{subequations}
\caption{\label{fig:CTs} \emph{Synthesis} part of the concurrency theorem.}
\end{figure}

\begin{figure}[ht!]
\begin{subequations}
\begin{align}
	\vcenter{\hbox{\includegraphics[scale=0.6,page=1]{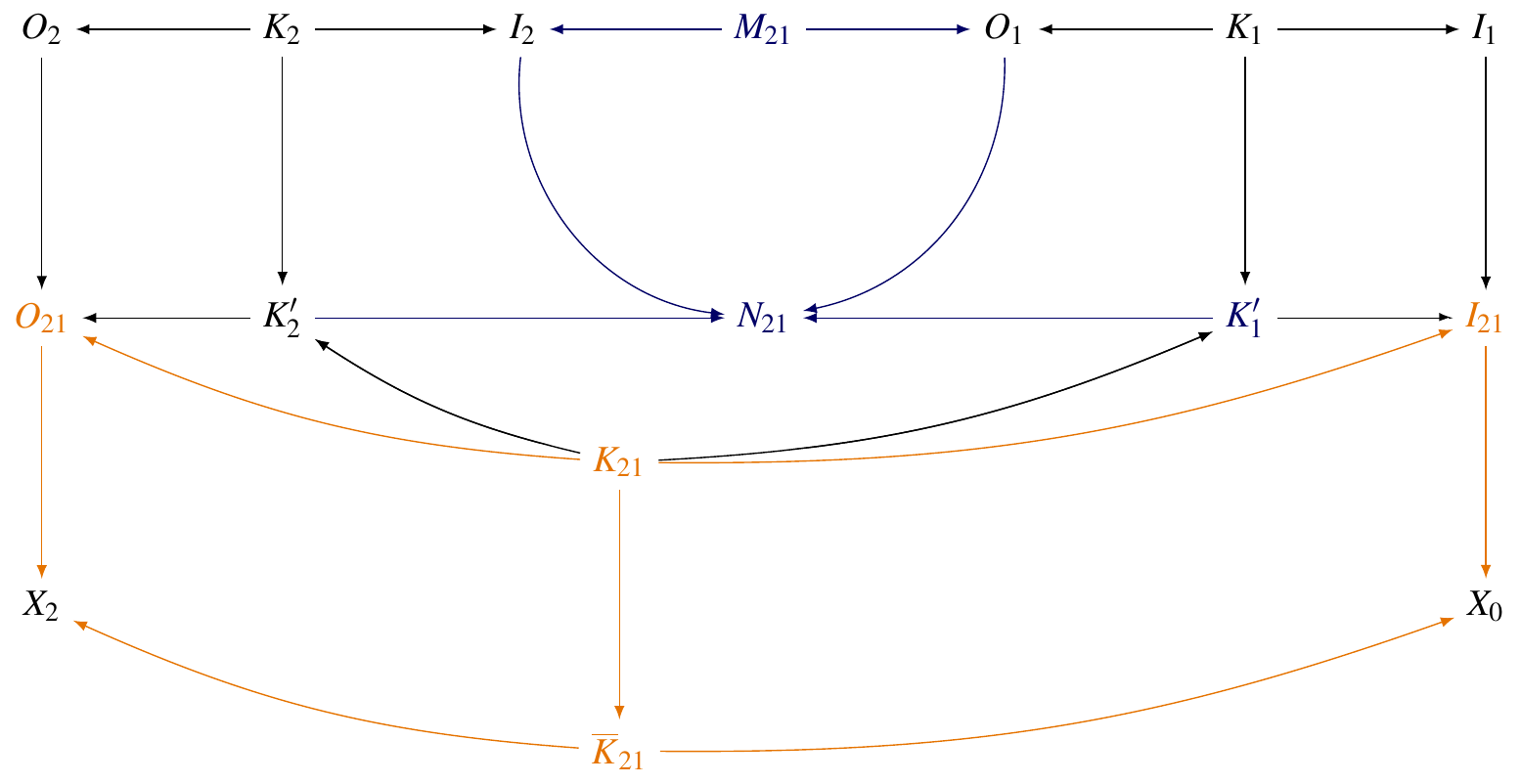}}}\label{eq:CTa1}\\
	\vcenter{\hbox{\includegraphics[scale=0.6,page=2]{images/concurrency-proof-Analysis.pdf}}}\label{eq:CTa2}\\
	\vcenter{\hbox{\includegraphics[scale=0.6,page=3]{images/concurrency-proof-Analysis.pdf}}}\label{eq:CTa3}\\
	\vcenter{\hbox{\includegraphics[scale=0.6,page=4]{images/concurrency-proof-Analysis.pdf}}}\label{eq:CTa4}
\end{align}
\end{subequations}
\caption{\label{fig:CTa} \emph{Analysis} part of the concurrency theorem.}
\end{figure}

\subsection{Proof of the SqPO associativity theorem}\label{app:SqPOassoc}

Our proof strategy will be closely related to the one presented in~\cite{bp2018} (with full technical details provided in~\cite{bp2019-ext}) for the analogous associativity theorem in the DPO-rewriting case. However, the SqPO-type case poses considerable additional challenges, since this rewriting semantics yields diagrams of a rather heterogeneous nature (including pullbacks, pushouts, pushout complements and FPCs) as compared to the DPO case, and in addition unlike DPO-type rule compositions, SqPO-type compositions are not reversible in general, which necessitates an independent proof of both directions of the bijective correspondence.
\begin{proof}
We first prove the claim in the ``$\Rightarrow$'' direction, i.e.\ starting from the set of data
\begin{equation}
\begin{aligned}
	{\color{h1color}\mathbf{m}_{21}}
		&=(O_2{\color{h1color}\leftarrow M_{21}\rightarrow }I_1)\in \sqRMatch{p_2}{p_1}\\
	{\color{h1color}\mathbf{m}_{3(21)}}
		&=(O_3{\color{h1color}\leftarrow M_{3(21)}\rightarrow }I_{21})\in \sqRMatch{p_3}{p_{21}}\,,\qquad p_{21}=\sqComp{p_2}{{\color{h1color}\mathbf{m}_{21}}}{p_1}\,,
\end{aligned}
\end{equation}
we have to demonstrate that one may uniquely (up to isomorphisms) construct from this information a pair of admissible matches
\begin{equation}
\begin{aligned}
	{\color{h1color}\mathbf{m}_{32}}
		&=(O_3{\color{h1color}\leftarrow M_{32}\rightarrow }I_2)\in \sqRMatch{p_3}{p_2}\\
	{\color{h1color}\mathbf{m}_{(32)1}}
		&=(O_{32}{\color{h1color}\leftarrow M_{(32)1}\rightarrow }I_{1})\in \sqRMatch{p_{32}}{p_1}\,,\qquad p_{32}=\sqComp{p_3}{{\color{h1color}\mathbf{m}_{32}}}{p_2}\,,
\end{aligned}
\end{equation}
and such that the property described in~\eqref{eq:THMassoc} holds. We begin by forming the SqPO-composite rule $p_{3(21)}=\sqComp{p_3}{\mathbf{m}_{3(21)}}{p_{21}}$, which results in the diagram
\begin{equation}\label{eq:AssocSqPOForwardStep1}
\vcenter{\hbox{\includegraphics[scale=0.5,page=1]{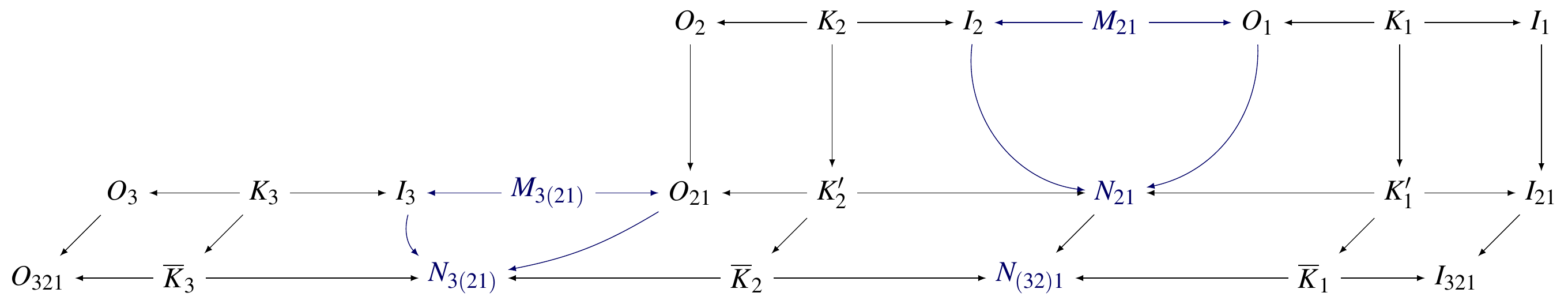}}}
\end{equation}
by virtue of invoking SqPO-composition twice. For the remainder of the proof, it is very important to precisely determine the nature of each of the squares in this diagram:
\begin{itemize}
	\item To clarify the structure of the rightmost four squares at the bottom, consider the setting presented in~\eqref{eq:CTa1}: since in the definition of the SqPO-composition as presented in~\eqref{eq:SqPOccomp} the nature of the squares to the right is that of pushout complement and pushout, respectively, it may be verified that applying the analysis procedure to the diagram in~\eqref{eq:CTa1} with thus the ``curvy'' front and right bottom faces both pushouts, one eventually arrives (by virtue of pushout-pushout and pushout-pullback decomposition) at the setting depicted in~\eqref{eq:CTa3} with all squares in the bottom row being pushouts. Thus the rightmost four squares at the bottom of~\eqref{eq:AssocSqPOForwardStep1} are all pushouts.
	\item By virtue of the definition of SqPO-composition, all vertical squares in the back of~\eqref{eq:CTa3} are pushouts, except for the square $\cSquare{K_2,K_2',{\color{h1color}N_{21}},I_2}$, which is an FPC. Analogously, the bottom leftmost three squares are (in order from left to right) a pushout, an FPC and a pushout.
\end{itemize}

Constructing the pullback ${\color{h1color}M_{32}}=\pB{{\color{h1color}M_{3(21)}\rightarrow}O_{21}\leftarrow O_2}$ (which by composition of pullbacks also leads to an arrow ${\color{h1color}M_{32}\rightarrow}I_3$) and forming the three additional vertical squares on the far left in the evident fashion in the diagram below
\begin{equation}
\vcenter{\hbox{\includegraphics[scale=0.5,page=2]{images/newAssocProof.pdf}}}
\end{equation}
allows us to construct ${\color{h1color}N_{32}}=\pO{I_3{\color{h1color}\leftarrow M_{32}\rightarrow}O_2}$, which in turn via universality of pushouts uniquely induces an arrow ${\color{h1color}N_{32}\rightarrow N_{3(21)}}$:
\begin{equation}
\vcenter{\hbox{\includegraphics[scale=0.5,page=3]{images/newAssocProof.pdf}}}
\end{equation}
Here, the rightmost three squares on the top are formed in the evident fashion (and they are pushouts by virtue of Lemma~\ref{lem:Main}\eqref{lem:fPOPB}), while the other arrows of the above diagram are constructed as follows: 
\begin{equation}
\begin{aligned}
	K_3'&=\pB{\overline{K}_3\rightarrow{\color{h1color}N_{3(21)}\leftarrow N_{32}}}\,,\quad & & 
	O_{32}&=\pO{K_3'\leftarrow K_3\rightarrow O_3}\\
	K_2''&=\pB{{\color{h1color}N_{32}\rightarrow N_{3(21)}\leftarrow}\overline{K}_2}\,,\quad & & 
	I_{32}&=\pO{K_2''\leftarrow K_2\rightarrow I_2}
\end{aligned}
\end{equation}
Invoking pushout-pullback, pushout-pushout and vertical FPC-pullback  decompositions, it may be verified that (describing positions of front and top square pairs by the position of the respective front square, from left to right)
\begin{itemize}
	\item the leftmost front and top squares are pushouts,
	\item the second from the left front and top squares are FPCs,
	\item the third from the left front and top squares are pushouts,
	\item in the next adjacent pair, the front square is an FPC and the top square a pushout,
	\item the second from the right front and top squares are pushouts, and
	\item the rightmost front and top squares are pushouts.
\end{itemize}

Defining the pullback object ${\color{h1color}M_{(32)1}}=\pB{I_{32}{\color{h1color}\rightarrow N_{3(21)}\leftarrow}O_1}$, thus inducing an arrow ${\color{h1color}M_{21}\rightarrow M_{3(21)}}$,
\begin{equation}
\vcenter{\hbox{\includegraphics[scale=0.5,page=4]{images/newAssocProof.pdf}}}
\end{equation}
it remains to verify that the square $\cSquare{{\color{h1color}M_{3(21)}},I_{32},{\color{h1color}N_{3(21)}},O_1}$ is not only a pullback, but also a pushout square. This part of the proof requires a somewhat intricate diagram chase; since the required arguments are identical\footnote{More precisely, the only difference in the structure of the relevant sub-diagram compared to the DPO case resides in the two FPC squares in the front and back in fourth position from the left (which happen to be pushout squares in the corresponding DPO-type proof), but the structure of this part of the diagram is not explicitly used in the proof in the DPO variant of the theorem, whence the claim follows.} to the ``$\Rightarrow$'' part of proof of DPO-type associativity as presented in~\cite{bp2019-ext}, we omit this part of the proof here in the interest of brevity.\\

It thus remains to prove the claim in the ``$\Leftarrow$'' direction, i.e.\ starting from the set of data
\begin{equation}
\begin{aligned}
	{\color{h1color}\mathbf{m}_{32}}
		&=(O_3{\color{h1color}\leftarrow M_{32}\rightarrow }I_2)\in \sqRMatch{p_3}{p_2}\\
	{\color{h1color}\mathbf{m}_{(32)1}}
		&=(O_{32}{\color{h1color}\leftarrow M_{(32)1}\rightarrow }I_{1})\in \sqRMatch{p_{32}}{p_1}\,,\qquad p_{32}=\sqComp{p_3}{{\color{h1color}\mathbf{m}_{32}}}{p_2}\,,
\end{aligned}
\end{equation}
we need to demonstrate that one may uniquely (up to isomorphisms) construct from this information a pair of admissible matches
\begin{equation}
\begin{aligned}
	{\color{h1color}\mathbf{m}_{21}}
		&=(O_2{\color{h1color}\leftarrow M_{21}\rightarrow }I_1)\in \sqRMatch{p_2}{p_1}\\
	{\color{h1color}\mathbf{m}_{3(21)}}
		&=(O_3{\color{h1color}\leftarrow M_{3(21)}\rightarrow }I_{21})\in \sqRMatch{p_3}{p_{21}}\,,\qquad p_{21}=\sqComp{p_2}{{\color{h1color}\mathbf{m}_{21}}}{p_1}\,,
\end{aligned}
\end{equation}
and such that the property described in~\eqref{eq:THMassoc} holds. We begin by forming the SqPO-composite rule $p_{(32)1}=\sqComp{p_{32}}{\mathbf{m}_{(32)1}}{p_{1}}$, which results in the diagram
\begin{equation}\label{eq:AssocSqPOReverseStep1}
\vcenter{\hbox{\includegraphics[scale=0.5,page=1]{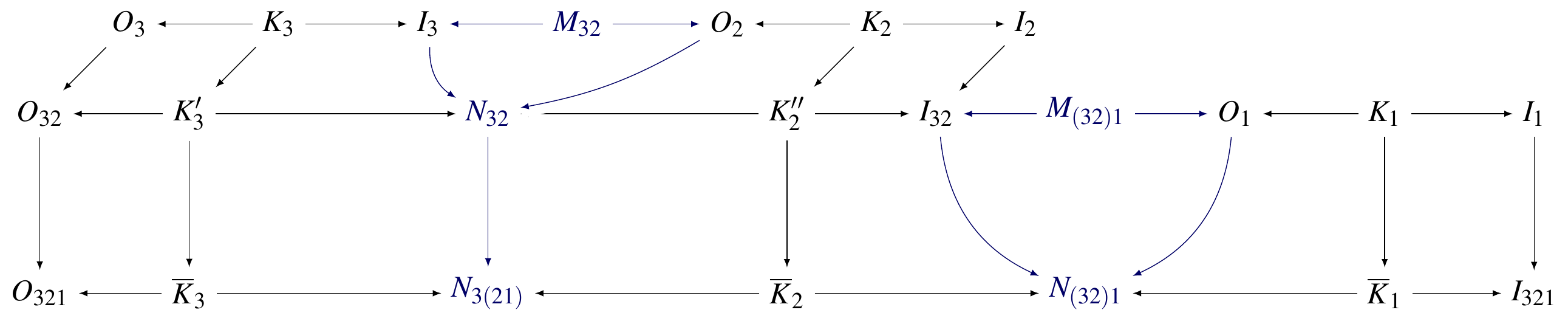}}}
\end{equation}
by virtue of invoking SqPO-composition twice. A careful inspection of the definition of the SqPO-composition and of the analysis part of the SqPO concurrency theorem permit to verify that the nature of all squares thus constructed coincides precisely with the nature of the corresponding squares in the ``$\Rightarrow$'' part of the proof.

Constructing the pullback ${\color{h1color}M_{21}}=\pB{I_2\rightarrow I_{32}\leftarrow\color{h1color}M_{(32)1}}$ (which by composition of pullbacks also leads to an arrow ${\color{h1color}M_{21}\rightarrow}O_1$) and forming the three additional vertical squares on the far right in the evident fashion in the diagram below
\begin{equation}
\vcenter{\hbox{\includegraphics[scale=0.5,page=2]{images/newAssocProof-SqPO-ReverseStep.pdf}}}
\end{equation}
allows us to construct ${\color{h1color}N_{21}}=\pO{I_2{\color{h1color}\leftarrow M_{21}\rightarrow}O_1}$, which in turn via the universal property of pushouts induces an arrow ${\color{h1color}N_{21}\rightarrow N_{(32)1}}$:
\begin{equation}
\vcenter{\hbox{\includegraphics[scale=0.5,page=3]{images/newAssocProof-SqPO-ReverseStep.pdf}}}
\end{equation}
The remaining new squares of the above diagram are constructed as follows:
\begin{equation}
\begin{aligned}
K_2'&=\pB{\overline{K}_2\rightarrow {\color{h1color}N_{(32)1}}\leftarrow {\color{h1color}N_{21}}} & & 
O_{21}&=\pO{O_2\leftarrow K_2\rightarrow K_2'}\,.
\end{aligned}
\end{equation}
Moreover, by virtue of vertical FPC composition, the square %
$\cSquare{K_3,\overline{K}_3,{\color{h1color}N_{3(21)}},I_3}$ is an FPC, while via pushout composition the square $\cSquare{K_3,\overline{K}_3,O_{321},O_3}$ is a pushout.\\

Again, the nature of all squares constructed thus far coincides precisely with the structure as presented in the ``$\Rightarrow$'' part of the proof, with one notable exception: by virtue of vertical FPC-pullback decomposition, we may only conclude that the square %
$\cSquare{K_2',\overline{K}_2,{\color{h1color}N_{(32)1}},{\color{h1color}N_{21}}}$ is an FPC (but at this point we do \emph{not} know whether it is also a pushout as in the analogous part of the diagram in the ``$\Rightarrow$'' part of the proof). However, an auxiliary calculation demonstrates that this square is in fact a pushout in disguise --- consider the following ``splitting'' of the relevant sub-part of the diagram as shown below:
\begin{equation}
\vcenter{\hbox{\includegraphics[scale=0.6,page=1]{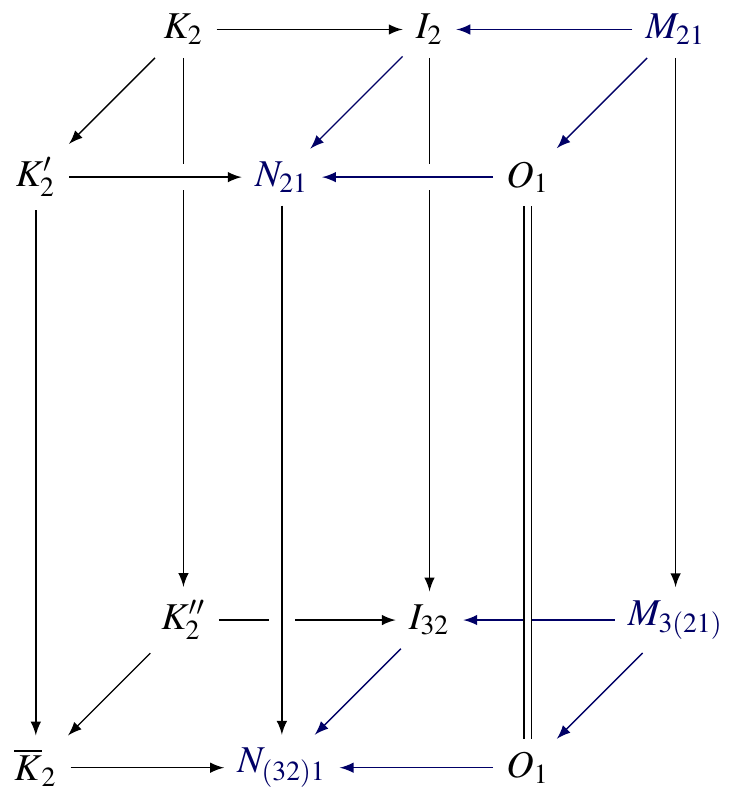}}}
\quad \rightsquigarrow\quad  \vcenter{\hbox{\includegraphics[scale=0.675,page=2]{images/new-assoc-proof-SqPO-reverse-Aux.pdf}}}
\end{equation}
The precise steps are as follows: the front left square in the diagram above left is an FPC; thus if one takes the pushout ${\color{h1color}N_{(32)1}'}=\pO{\overline{K}_2\leftarrow K_2'\rightarrow {\color{h1color}N_{21}}}$ as well as in the bottom back the pushout along the isomorphism of $K_2''$ as displayed (yielding the arrows in the middle), followed by taking the pullback $O_1'=\pB{{\color{h1color}N_{(32)1}'},{\color{h1color}N_{(32)1}}, O_1}$ (which entails that the arrow ${\color{h1color}M_{3(21)}}\rightarrow O_1'$ exists), it is straightforward to verify that $O_1'\cong O_1$. Invoking the van Kampen property (recalling that by definition the right square on the bottom is a pushout), we find that $\cSquare{{\color{h1color}M_{3(21)}},O_1',{\color{h1color}N_{(32)1}'},I_{32}}$ is a pushout. Thus by pushout-pushout decomposition, the square $\cSquare{O_1',O_1,{\color{h1color}N_{(32)1}},{\color{h1color}N_{(32)1}'}}$ is a pushout, whence ${\color{h1color}N_{(32)1}'}\cong {\color{h1color}N_{(32)1}}$. This in summary entails\footnote{Coincidentally, at this point we are back into full structural analogy to the ``$\Rightarrow$'' part of the proof, a necessary prerequisite for completing this part of the proof as it will turn out.} that the square $\cSquare{K_2',\overline{K}_2,{\color{h1color}N_{(32)1}},{\color{h1color}N_{21}}}$ is not only an FPC, but in fact also a pushout.\\

Back to the main proof, defining the pullback object 
\[
{\color{h1color}M_{3(21)}}=\pB{I_{3}{\color{h1color}\rightarrow N_{(32)1}\leftarrow}O_{21}}\,,
\]
thus inducing an arrow ${\color{h1color}M_{32}\rightarrow M_{3(21)}}$,
\begin{equation}
\vcenter{\hbox{\includegraphics[scale=0.5,page=4]{images/newAssocProof-SqPO-ReverseStep.pdf}}}
\end{equation}
it remains to verify that the square $\cSquare{{\color{h1color}M_{3(21)}},O_{21},{\color{h1color}N_{3(21)}},I_3}$ is not only a pullback, but also a pushout square. Let us construct the auxiliary diagram as depicted in Figure~\ref{fig:SqPOrevAuxDiag}, with objects obtained via taking suitable pullbacks as indicated. The four cubes that are drawn separately are the top, back, bottom and front cubes induced via the newly constructed arrows, and are oriented such that one may easily apply the van Kampen property in the next step of the proof (which in most cases requires a suitable 3d-rotation).
\afterpage{%
  \begin{landscape}
    \begin{figure}%
    \centering
\includegraphics[scale=0.6]{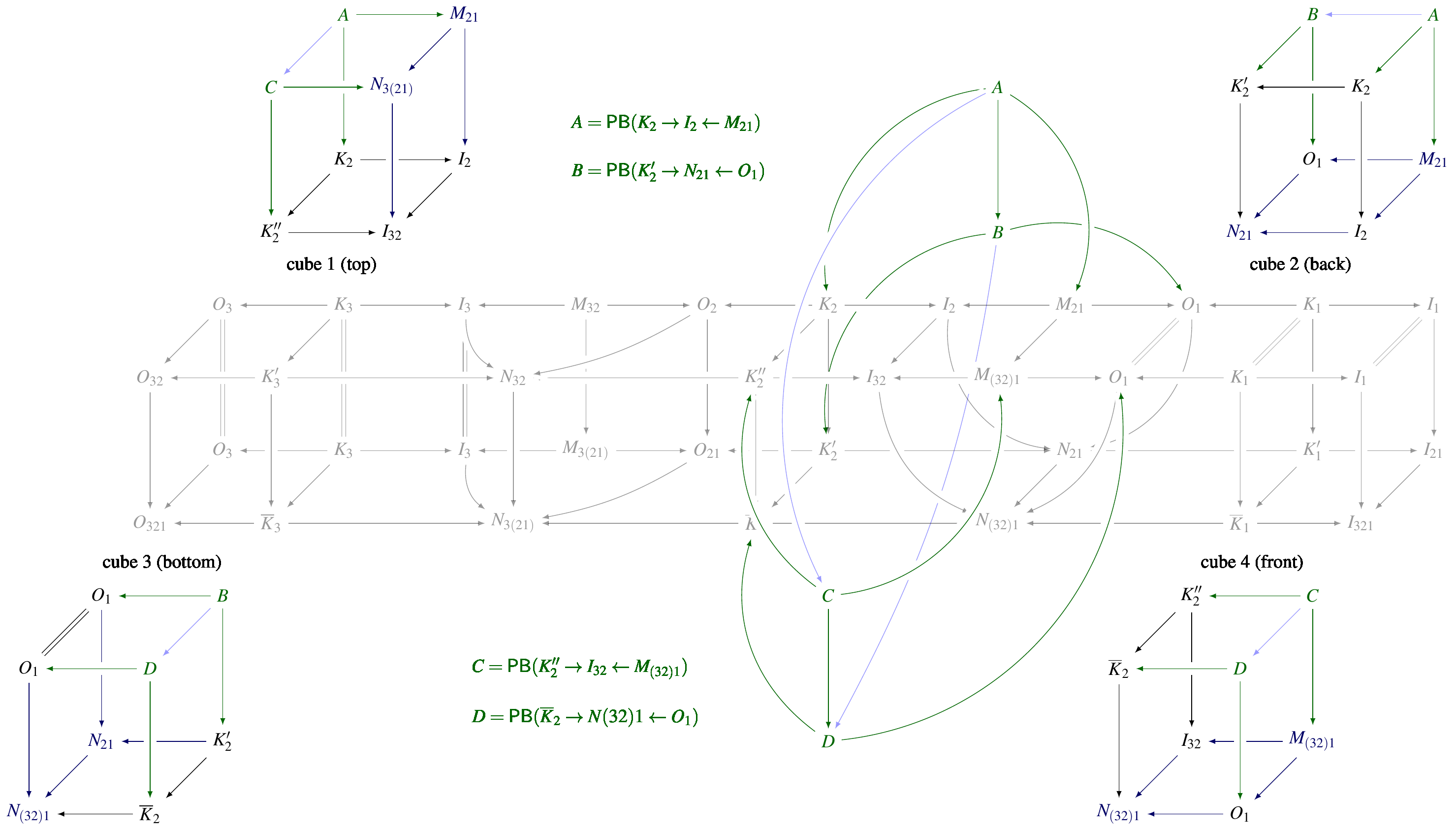}
\caption{\label{fig:SqPOrevAuxDiag}Auxiliary diagram for the second part of the SqPO associativity proof.}
    \end{figure}
  \end{landscape}%
}

Invoking pullback-pullback decomposition and the van Kampen property repeatedly, it may be verified that in the relevant sub-diagram as presented below
\begin{equation}
\vcenter{\hbox{\includegraphics[scale=0.75,page=1]{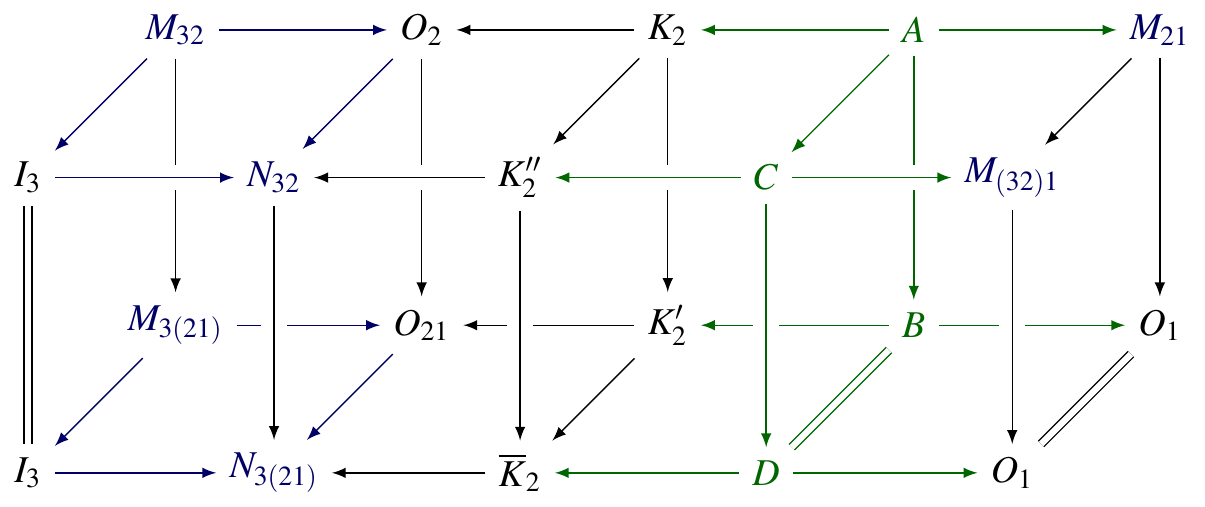}}}
\end{equation}
we find the following structure of the squares:
\begin{itemize}
	\item All squares on the top are pushouts, except the second one from the right (which is a pullback).
	\item The second and third square from the left in the back of the diagram are pushouts, the other two back squares are pullbacks, with the same structure for the front squares.
	\item Counting from left to right, the second and fourth square on the bottom are pushouts, the other two are pullbacks.
\end{itemize}
In particular, as indicated this entails that $D\cong B$. We proceed by performing the following ``splitting'' of the diagram:
\begin{equation}
\vcenter{\hbox{\includegraphics[scale=0.75,page=2]{images/newAssocProof-Step3-SqPO-reverse.pdf}}}
\end{equation}
Start the ``splitting'' via taking the pushouts ${\color{h1color}N_{32}'}=\pO{I_3\leftarrow I_3\rightarrow {\color{h1color}N_{32}}}$ (which entails that ${\color{h1color}N_{32}'}\cong {\color{h1color}N_{32}}$) and $O_2'=\pO{{\color{h1color}M_{3(21)}}\leftarrow {\color{h1color}M_{32}}\rightarrow O_2}$. By pullback-pullback decomposition followed by pushout-pullback decomposition, we may conclude that the resulting square $\cSquare{O_2',{\color{h1color}N_{32}'},I_3,{\color{h1color}M_{3(21)}}}$ is a pushout. Note also that all vertical squares in the bottom left part of the diagram thus constructed are pullbacks (by virtue of suitable pullback decompositions).

Next, construct the pullbacks $\overline{K}_2'=\pB{{\color{h1color}N_{32}'\rightarrow N_{3(21)}}\leftarrow \overline{K}_2}$ and $K_2'''=\pB{O_2'\rightarrow O_{21}\leftarrow K_2'}$. By pushout-pullback decomposition, in the diagram on the right the top and bottom front squares and back squares in the second column are pushouts, while the square $\cSquare{K_2''',K_2',\overline{K}_2,\overline{K}_2'}$ is a pullback. Performing the precise same steps in the next column, i.e.\ via taking the pullbacks $D'=\pB{\overline{K}_2'\rightarrow \overline{K}_2\leftarrow D}$ and $B'=\pB{K_2'''\rightarrow K_2'\leftarrow B}$, we obtain via pushout-pullback decomposition that the top and bottom front and top and bottom back squares in the third column are pushouts, while the square $\cSquare{B',B,D,D'}$ is a pullback. But since $D\cong B$ and since isomorphisms are stable under pullback, we conclude that $D'\cong B'$, and thus that the square $\cSquare{B',B,D,D'}$ is in fact a pushout. 

By pushout composition, we may thus conclude that $\cSquare{K_2''',K_2',\overline{K}_2,\overline{K}_2'}$ is a pushout, whence the square $\cSquare{O_2',O_{21},{\color{h1color}N_{3(21)}},{\color{h1color}N_{32}'}}$ is a pushout, which finally allows us to verify that $\cSquare{{\color{h1color}M_{3(21)}},O_{21},{\color{h1color}N_{3(21)}},I_3}$ is a pushout. This concludes the proof of the SqPO-type associativity theorem.
\end{proof}

\subsection{Proof of the theorem on SqPO rule algebra properties}\label{app:SqPOraProps}

\begin{proof}
Associativity of $\odot_{\cR_{\bfC}}$ follows from the associativity of the operation $\sqComp{.}{.}{.}$ proved in Theorem~\ref{thm:SqPOassoc}. The claim that $R_{\inObj}=\delta(\inObj\leftarrow\inObj\rightarrow \inObj)$ is the unit element of the rule algebra $\cR^{sq}_{\bfC}$ follows directly from the definition of the rule algebra product for $R_{\inObj}\odot_{\cR_{\bfC}}R$ and $R\odot_{\cR_{\bfC}}R_{\inObj}$ for $R\in \cR^{sq}_{\bfC}$. More concretely, we present below the category-theoretic composition calculation that underlies the equation $R_{\inObj}\odot_{\cR_{\bfC}}R=R$: 
\begin{equation}\gdef\mycdScale{0.85}
\begin{mycd}
			\inObj\ar[d] & 
			\inObj 
				\ar[l]\ar[r]\ar[d]
				\ar[dl,phantom,"\mathsf{PO}"] & 
			\inObj 
				\ar[dr,h1color,bend right]\ar[dl,phantom,"\mathsf{FPC}"] & 
			{\color{h1color}\inObj}
				\ar[l,h1color]\ar[r,h1color]\ar[d,h1color,phantom,"\mathsf{PO}"] & 
			O \ar[dl,h1color,bend left]\ar[dr,phantom,"\mathsf{POC}"]&
			K \ar[l,"o"']\ar[r,"i"]\ar[d]
				\ar[dr,phantom,"\mathsf{PO}"] & 
			I\ar[d]\\
			{\color{h2color}O} & 
			O\ar[l]\ar[rr] & & 
			{\color{h1color}O} & &
			K\ar[ll,"o"]\ar[r,"i"'] & {\color{h2color}I}\\
		\end{mycd}
\end{equation}
Here, it is important to note that the pushout complement used to construct the square marked $\mathsf{POC}$ always exists (see Lemma~\ref{lem:Main}\eqref{lem:fPOPB}), whence the claim follows.
\end{proof}

\subsection{Proof of the SqPO canonical representation theorem}\label{app:SqPOcanrep}

\begin{proof}
In order for $\rho^{sq}_{\bfC}$ to qualify as an algebra homomorphism (of unital associative algebras $\cR^{sq}_{\bfC}$ and $End(\hat{\bfC})$), we must have (with $R_{\inObj}=\delta(p_{\inObj})$, $p_{\inObj}=(\inObj\leftarrow \inObj\rightarrow \inObj)$)
\begin{align*}
&(i)\; \rho^{sq}_{\bfC}(R_{\inObj})=\mathbb{1}_{End(\hat{\bfC})}\\
&(ii)\;\forall R_1,R_2\in \cR^{sq}_{\bfC}:\; \rho^{sq}_{\bfC}(R_1\odot_{\cR_{\bfC}}R_2)=\rho^{sq}_{\bfC}(R_1)\rho^{Sq}_{\bfC}(R_1)\,.
\end{align*}
Due to linearity, it suffices to prove the two properties on basis elements $\delta(p),\delta(q)$ of $\cR^{sq}_{\bfC}$ (for $p,q\in \Lin{\bfC}$) and on basis elements $\ket{X}$ of $\hat{\bfC}$. Property $(i)$ follows directly from the definition,
\[
\forall X\in \obj{\bfC}_{\cong}:\quad \rho^{sq}_{\bfC}(R_{\inObj})\ket{X}\overset{\eqref{eq:canRepSqPO}}{=}\sum_{m\in \sqMatch{r_{\inObj}}{X}}\ket{(r_{\inObj})_m(X)}=\ket{X}\,.
\]
Property $(ii)$ follows from Theorem~\ref{thm:SqPOconcur} (the SqPO-type concurrency theorem): for all basis elements $\delta(p),\delta(q)\in \cR^{sq}_{\bfC}$ (with $p,q\in Lin(\bfC)$) and for all $X\in \obj{\bfC}$,
\begin{align*}
\rho^{sq}_{\bfC}\left(\delta(q)\odot_{\bfC}\delta(p)\right)\ket{X}
&\overset{\eqref{eq:defRcompSqPO}}{=}
\sum_{\mathbf{d}\in \sqRMatch{q}{p}}
\rho^{sq}_{\bfC}
	\left(
		\delta\left(
			\sqComp{q}{\mathbf{d}}{p}
			\right)
		\right)\ket{X}\\
&\overset{\eqref{eq:canRepSqPO}}{=}
\sum_{\mathbf{d}\in \sqRMatch{q}{p}}\;
\sum_{e\in \sqMatch{r_{\mathbf{d}}}{X}}
\ket{(r_{\mathbf{d}})_e(X)}
 \quad &(r_{\mathbf{d}}=\sqComp{q}{\mathbf{d}}{p})\\
&=
\sum_{m\in \sqMatch{p}{X}}
\sum_{n\in\sqMatch{q}{p_{m}(X)}}
\ket{q_n(p_m(X))} &\text{(via Thm.~\ref{thm:SqPOconcur})}\\
&\overset{\eqref{eq:canRepSqPO}}{=}
\sum_{m\in \sqRMatch{p}{X}}
\rho^{sq}_{\bfC}\left(\delta(q)\right)\ket{p_m(X)}\\
&\overset{\eqref{eq:canRepSqPO}}{=}
\rho^{sq}_{\bfC}
\left(\delta(q)\right)
\rho^{sq}_{\bfC}\left(\delta(p)\right)\ket{X}\,.
\end{align*}%
\end{proof}

\subsection{Proof of the SqPO stochastic mechanics framework theorem}\label{app:StochMechProof}\label{sec:appProofSMF}

\begin{proof}
By definition, the SqPO-type canonical representation of a generic rule algebra element $(\grule{O}{p}{I})\in \cR_{\bfC}$ is a row-finite linear operator, since by virtue of the finitarity of objects according to Assumption~\ref{ass:RAsqpo}, for every object $X\in \obj{\bfC}$ the set of SqPO-admissible matches $\sqMatch{p}{X}$ of the associated linear rule $p=(O\xleftarrow{o}K\xrightarrow{i}I)$ is finite. We may thus verify that the linear operator $H$ possesses all the required properties of a so-called $Q$-matrix (or infinitesimal generator) of a CTMC~\cite{norris,Anderson_1991}, i.e.\ its non-diagonal entries are non-negative, its diagonal entries are finite, and furthermore the row sums of $H$ are zero (whence $H$ constitutes a conservative and stable $Q$-matrix; compare~\eqref{def:Hprops} of Definition~\ref{def:CTMCs}). It is crucial to note that while originally $H$ as a linear combination of representations of rule algebra elements is only defined to act on \emph{finite} linear combinations of basis vectors $\ket{X}$ of $\hat{\bfC}$, an important mathematical result from the theory of CTMCs entails that if a row-finite linear operator such as $H$ is a stable and conservative $Q$-matrix, it extends to a linear operator on \emph{infinitely supported distributions} (here over basis vectors of $\hat{\bfC}$) with finite real coefficients (see e.g.\ \cite{Anderson_1991}, Chapters~1 and~2). Moreover, the property $\bra{}H=0$ follows directly from the defining equations~\eqref{def:Hprops} of Definition~\ref{def:CTMCs}.

Let us prove next the claim on the precise structure of observables. Recall that according to Definition~\ref{def:obs}, an observable $O\in \cO_{\bfC}$ must be a linear operator in $End(\cS_{\bfC})$ that acts diagonally on basis states $\ket{X}$ (for $X\in \obj{\bfC}_{\cong}$), whence that satisfies for all $X\in \obj{\bfC}_{\cong}$
\[
O\ket{X}=\omega_O(X)\ket{X}\quad (\omega_O(X)\in \bR)\,.
\]
Comparing this equation to the definition of the SqPO-type canonical representation (Definition~\ref{def:canRepSqPO}) of a generic rule algebra basis element $\delta(p)\in \cR^{sq}_{\bfC}$ (for $p\equiv(O\xleftarrow{o}K\xrightarrow{i}I)\in \Lin{\bfC}$),
\[
\rho^{sq}_{\bfC}(\delta(p))\ket{X}:=\begin{cases}
\sum_{m\in \sqMatch{p}{X}}\ket{p_m(X)}\quad &\text{if }\sqMatch{p}{X}\neq \emptyset\\
0_{\hat{\bfC}}&\text{else,}
\end{cases}
\]
we find that in order for $\rho^{sq}_{\bfC}(\delta(p))$ to be diagonal we must have
\[
 \forall X\in \obj{\bfC}:\forall m\in \sqMatch{p}{X}:\quad p_m(X)\cong X\,.
\]
But by definition of SqPO-type derivations of objects along admissible matches (Definition~\ref{def:SqPOr}), the only linear rules $p\in \Lin{\bfC}$ that have this special property are precisely the rules of the form
\[
p_{id_M}= (M\xleftarrow{id_M}M\xrightarrow{id_M}M)\,.
\]
In particular, defining $O_M^{sq}:=\rho^{sq}_{\bfC}(\delta(p_{id_M}))$, we find that the eigenvalue $\omega_{O_M^{sq}}(X)$ coincides with the cardinality of the set $\sqMatch{p_{id_M}}{X}$ of SqPO-admissible matches,
\[
\forall X\in ob(\bfC):\quad  O_M^{sq}\ket{X}=|\sqMatch{p_{id_M}}{X}|\cdot\ket{X}\,.
\]
This proves that the operators $O^{sq}_M$ form a basis of diagonal operators on $End(\hat{\bfC})$ (and thus on $End(\cS_{\bfC})$) that can arise as linear combinations of representations of rule algebra elements. 

To prove the jump-closure property, note that it follows from Definition~\ref{def:SqPOr} that for an arbitrary \emph{linear} rule $p\equiv(O\xleftarrow{o}K\xrightarrow{i}I)\in \Lin{\bfC}$, a generic object $X\in \obj{\bfC}$ and a monomorphism $m:I\rightarrow X$, $m$ is according to Definition~\ref{def:SqPOr} both a match of the rule $p$ as well as of the rule $p_{id_I}$. Evidently, the application of the rule $p$ to $X$ along the match $m$ produces an object $p_m(X)$ that is in general different from the object $p_{id_{I_m}}(X)$ produced by application of the rule $p_{id_I}$ to $X$ along the match $m$. But by definition of the projection operator $\bra{}$ (Definition~\ref{def:obs}),
\[
\forall X\in \obj{\bfC}_{\cong}:\quad \braket{}{X}:=1_{\bR}\,,
\]
we find that
\[
	\braket{}{p_m(X)}=\braket{}{p_{id_{I_m}}(X)}=1\,,
\]
whence we may prove the claim of the SqPO-type jump-closure property via verifying it on arbitrary basis elements (with notations as above):
\[
\bra{}\rho^{sq}_{\bfC}(\delta(p))\ket{X}
=|\sqMatch{p}{X}|=|\sqMatch{p_{id_I}}{X}|
=\bra{}\rho^{sq}_{\bfC}(\delta(p_{id_I}))\ket{X}\,.
\]
Since $X\in \obj{\bfC}_{\cong}$ was chosen arbitrarily, we thus have indeed that
\[
\bra{}\rho^{sq}_{\bfC}(\delta(p))=\bra{}\rho^{sq}_{\bfC}(\delta(p_{id_I}))\,.
\]
This concludes the proof that our definition of continuous-time Markov chains based upon SqPO-type rewriting rules is well-posed and yields all the requisite properties.
\end{proof}

\end{document}